\documentclass{article}

\usepackage[preprint]{neurips_2026}

\usepackage[utf8]{inputenc} 
\usepackage[T1]{fontenc}    
\usepackage{hyperref}       
\usepackage{url}            
\usepackage{booktabs}       
\usepackage{amsfonts}       
\usepackage{nicefrac}       
\usepackage{microtype}      
\usepackage{xcolor}         
\usepackage{graphicx}
\usepackage{amsmath}
\usepackage{amssymb}
\usepackage{mathtools}
\usepackage{amsthm}
\usepackage{multirow}
\usepackage{algorithm}
\usepackage{algorithmic}
\usepackage{dsfont}
\usepackage{tikz}
\usetikzlibrary{positioning, arrows.meta, calc, decorations.pathmorphing, backgrounds, fit}

\theoremstyle{plain}
\newtheorem{theorem}{Theorem}[section]

\newtheorem{lemma}[theorem]{Lemma}
\newtheorem{corollary}[theorem]{Corollary}

\theoremstyle{definition}

\newtheorem{definition}[theorem]{Definition}

\usepackage{wrapfig}
\usepackage{caption}
\usepackage{subcaption}
\makeatletter
\renewcommand{\ALG@name}{Alg.}
\makeatother
\usepackage[most]{tcolorbox}  

\usepackage{soul}
\title{Context-Gated Associative Retrieval: From Theory to Transformers}

\author{Moulik Choraria$^{1,}$ $\thanks{$^{1}$Equal contribution; Correspondence: moulikc2@illinois.edu}$ \\
UIUC \\
\And
Argyrios Gerogiannis$^{1}$ \\
UIUC \\
\And
Vidhata Jayaraman \\
UIUC \\
\AND
Ankur Mani \\
University of Minnesota, Twin Cities
\And
Lav R. Varshney \\
Stony Brook University 
}

\begin{document}

\maketitle

\begin{abstract}

Hopfield networks and their generalizations have established deep connections among biological associative memories, statistical physics, and transformers. Yet most models treat retrieval as a fixed query-to-memory mapping, ignoring the role of external context in recall. In this work, we propose a two-stage associative memory architecture, wherein a context-gate subcircuit reshapes the retrieval energy landscape before and during recall. We show theoretically that context gating increases inter-memory separation while inducing sparsity, translating into exponential improvements in retrieval. Crucially, we prove that the system admits a unique self-consistent fixed point, revealing that the resulting retrieval state is driven by both a direct contextual bias and a second-order retrieval-gate feedback loop. We then bridge this theory to transformers; specifically, we evaluate a first-order approximation on Llama-3, confirming that in-context learning acts as context-gated retrieval. Native dynamics mirror our theory: context localizes a memory subspace, enabling the zero-shot query to cleanly discriminate. Ultimately, this framework provides a mechanistic link between associative memory theory and LLM phenomenology.
\end{abstract}

\section{Introduction}

The Hopfield network (\cite{hopfield1982, hopfield1984}) established one of the earliest bridges between statistical physics and neural computation, casting memory as energy minimization whose attractors encode stored patterns. Recently, dense associative memories (\cite{krotov2016dense, demircigil2017model}) have allowed lifting the classical linear capacity to polynomial scaling via higher-order interactions, and modern Hopfield networks (MHNs) (\cite{ramsauer2020hopfield}), achieve exponential capacity while recovering softmax attention as the update rule---revealing an equivalence between associative recall and transformer attention (\cite{vaswani2017attention})
This connection has inspired architectures and new insights into in-context learning (ICL) (\cite{burnsFE2025, wuHHZL2025}), factual recall (\cite{nichaniLB2024}), and energy-based generative modeling (\cite{phamRNZAK2025}).

What remains less understood is how $\emph{external context}$ shapes retrieval. Biological memory is modulated by behavioral state and task demands (\cite{smithV2001}), yet most associative memory models treat the query as the sole retrieval input. This context-dependence phenomenon also manifests in large language models (LLMs), for instance, when the $\emph{same}$ query elicits drastically different behaviors based on the system prompt (\cite{nardo2023waluigi}), or as ICL demonstrably improves with additional contextual examples (\cite{brownMRSKDNSS_etal_2020}). Interpretability-adjacent work (\cite{hendelGG2023, merulloEP2024, lvCZWLWXY2024}) further suggests a two-phase delineation in tasks like ICL or factual recall: (a) a $\emph{function determination}$ phase that only processes context to select the task (e.g., $\textit{get\_capital()})$, leaving the query untouched, and (b) an $\emph{execution}$ phase that retrieves the answer for the query ($\textit{get\_capital(France) = Paris}$), while evidence of similar redundancies in multimodal processing (\cite{JiangCZLSY2025, ChorariaWBSWZSV2025}) hint at a similar phenomenon at play.  

Recent work has begun addressing this gap: \cite{bettetiBBZ2025} proposed input-driven synaptic plasticity that reshapes Hopfield energy landscapes, while \cite{podlaskiAV2025} introduced context-modular networks that partition binary memories across contexts. However, a unified energy-based architecture modeling this continuous interaction is still missing. To address this, we propose a \textbf{two-stage associative memory} within the framework of \cite{krotov2021ham}, where an external context drives query-based retrieval. Our main contributions are as follows:

\begin{enumerate}
    \item \textbf{Theoretical Separation and Sparsity:} We show analytically that context gating increases the effective separation between target and competitor memories, translating into exponential improvements in retrieval via the MHN framework. Furthermore, we characterize a phase transition in the isolated gate dynamics that natively induces winner-take-all sparsity.
    \item \textbf{Self-Consistent Contextual Retrieval:} We formalize context gating and retrieval as a coupled dynamical system. We then establish that this system admits a unique self-consistent fixed point, revealing that recall is driven by both a direct context-filtered bias and non-linear retrieval-gate feedback.
    \item \textbf{Bridging Theory to Transformers:} We translate our architectural framework to the internal mechanisms of LLMs. By evaluating a first-order approximation of our model on Llama-3-8B, we demonstrate that native transformer dynamics functionally mirror our theory during in-context learning and factual recall: context first localizes the task-relevant memory subspace, enabling the zero-shot query to cleanly discriminate within this narrowed set.
\end{enumerate}

\section{Related Work}
\textbf{Biological Parallels: Astrocytes as Contextual Gates.} Standard associative memories rely on strictly neuron-to-neuron interactions, but biological memory leverages multi-cellular architectures to integrate external context. Notably, astrocytes act as active, context-dependent gatekeepers of neural activity (\cite{ValtchevaV2016}). Modulated by global behavioral states (\cite{LeftonWDOZ_et_al2025, GuttenplanMSBMAKF2025}), astrocytic networks structurally prime local synaptic landscapes, dictating neuronal responsiveness and plasticity. This maps directly onto our two-stage architecture: our context-gate sub-circuit serves as the astrocytic analogue, processing external context to dynamically reshape the retrieval energy landscape prior to recall.

\textbf{Hopfield Networks \& Sparsity.} Dense Associative Memories and Hopfield Networks \citep{karbasiSSV2013, krotov2016dense, demircigil2017model, ramsauer2020hopfield} achieve exponential storage capacity, recovering dense softmax attention as their continuous update rule. However, because softmax assigns global non-zero probabilities, MHNs often retrieve a continuous superposition of memories. To improve retrieval speed and accuracy, Sparse Modern Hopfield Networks (SMHNs) \citep{HuYWXCL2023, SantosNMM2025} induce sparsity by modifying the energy function and update rules. While our gating mechanism similarly induces sparse retrieval, it fundamentally diverges from these approaches: rather than altering the $\textrm{LogSumExp}$ energy function, our framework enforces sparsity architecturally by narrowing the searchable subspace before the query is applied.

\textbf{Steering Vectors \& In-Context Learning.} Mechanistic interpretability suggests LLMs execute in-context learning (ICL) by compressing task demonstrations into task or function vectors (\cite{hendelGG2023, toddLSMWB2024}). Building on this, \cite{liuYXZ2024} show that extracting an "In-Context Vector" (ICV) from demonstrations and injecting it into the residual stream steers the model to execute intended tasks on new queries. In our framework, the ICV directly parallels our external context vector: it reshapes the energy landscape by encoding the task, localizing the relevant memory subspace for the query to act upon. This yields a principled associative memory perspective on empirical ICL steering phenomena.

\section{Theory}

\subsection{Construction}
\label{sec:const}

We design a two-stage associative memory with the goal of exploiting the role of context in retrieval. For brevity, we defer a formal primer on the modular energy framework~(\cite{krotovHPP2025}) to Appendix~\ref{sec:appendix_primer_gdam}. Here, we describe its architecture and key components. 

\subsubsection{Neuron Layers}
The architecture (Figure~\ref{fig:arch}) comprises four layers. The \textbf{context} ($c \in \mathbb{R}^{d_c}$) receives external input specifying task or behavioral state, the \textbf{gates} ($s \in \mathbb{R}^{N}$) which contain one neuron per stored memory and selectively activate a sparse subset based on context, and the \textbf{query} ($q \in \mathbb{R}^{d_q}$) layer, which receives the noisy input for retrieval. Each uses a quadratic Lagrangian $L(x) = \tfrac{1}{2}\|x\|^2$, so the activation is the identity: $\hat{c} = c$, $\hat{s} = s$, $\hat{q} = q$. On the other hand, the \textbf{retrieval} layer ($r \in \mathbb{R}^{N}$) uses a LogSumExp Lagrangian, $L_r(r) = \tfrac{1}{\beta}\log\sum_{i=1}^{N} \exp(\beta r_i)\,$, yielding temperature-scaled softmax activations $\hat{r}_i = \exp(\beta r_i)/\sum_j \exp(\beta r_j)$. Since $\hat{r}_i \in [0,1]$ and $\sum_i \hat{r}_i = 1$, this layer naturally produces a distribution over the $N$ stored memories.

\begin{figure}[t]
\centering
\resizebox{0.4\columnwidth}{!}{%
\begin{tikzpicture}[
    node distance=2.8cm and 2.8cm,
    layer/.style={
        circle, draw, thick, minimum size=1.2cm,
        font=\sffamily\bfseries\Large
    },
    context/.style={layer, fill=blue!10, draw=blue!50},
    gate/.style={layer, fill=orange!12, draw=orange!60},
    retrieval/.style={layer, fill=green!10, draw=green!50!black},
    query/.style={layer, fill=red!8, draw=red!50},
    synlabel/.style={font=\sffamily\small, fill=white, inner sep=2pt, rounded corners=2pt},
    actlabel/.style={font=\sffamily\footnotesize, text=gray!70},
    groupbox/.style={draw=gray!40, dashed, rounded corners=10pt, inner sep=14pt},
]
    \node[gate] (S) {$S$};
    \node[context, below=1.6cm of S] (C) {$C$};
    \node[retrieval, right=of S] (R) {$R$};
    \node[query, below=1.6cm of R] (Q) {$Q$};

    \begin{scope}[on background layer]
        \node[groupbox, fill=blue!2, fit=(S)(C)] (ctxbox) {};
        \node[groupbox, fill=red!2, fit=(R)(Q)] (retbox) {};
    \end{scope}

    \node[font=\sffamily\footnotesize, text=gray!60, anchor=south] at (ctxbox.north) {Context Circuit};
    \node[font=\sffamily\footnotesize, text=gray!60, anchor=south] at (retbox.north) {Retrieval Circuit};

    \node[actlabel, anchor=north] at ($(C.south)+(0,-0.05)$) {$\hat{c} = c$};
    \node[actlabel, anchor=east] at ($(S.west)+(-0.2,0)$) {$\hat{s} = s$};
    \node[actlabel, anchor=west] at ($(R.east)+(0.2,0)$) {$\hat{r}\!=\!\mathrm{softmax}(\beta r)$};
    \node[actlabel, anchor=north] at ($(Q.south)+(0,-0.05)$) {$\hat{q} = q$};

    \draw[<->, thick, blue!60] (S) -- (C)
        node[midway, synlabel, left=5pt] {$h_1$};

    \draw[->, thick, orange!60] ($(S.north)+(0.15,0)$) arc[start angle=0, end angle=320, radius=0.45cm]
        node[synlabel, above=6pt, xshift=-2pt] at (S.north) {$h_2$};

    \draw[<->, thick, purple!50, densely dashed] (S) -- (R)
        node[midway, synlabel, above=5pt] {$h_3\;(\lambda)$};

    \draw[<->, thick, green!50!black] (R) -- (Q)
        node[midway, synlabel, right=5pt] {$h_4$};

\end{tikzpicture}}
\caption{Two-stage associative memory architecture. The context circuit (left) establishes gate activations $s$ from context $c$ via alignment ($h_1$) and inter-memory competition ($h_2$). The retrieval circuit (right) recalls from stored memories via softmax attention ($h_4$). Cross-circuit coupling ($h_3$, dashed) enables bidirectional information flow controlled by $\lambda$.}
\label{fig:arch}
\vspace{-0.7em}
\end{figure}
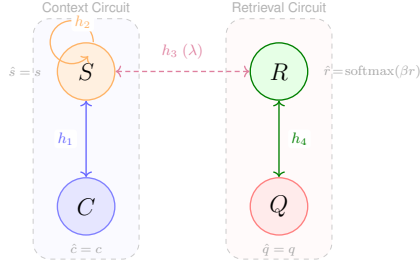

\subsubsection{Hypersynapses}

We next describe the four hypersynapses encoding the functional structure of the network.

\textbf{(i) Context-gate alignment ($h_1$).}
A bilinear energy $E_{h_1} = -\hat{s}^\top W_{h_1} \hat{c}$ which promotes alignment between the gate state and the context signal. The weight matrix $W_{h_1} \in \mathbb{R}^{N \times d_c}$ encodes how context activates gates, which in turn drive retrieval. This produces synaptic currents $I_s^{(h_1)} = W_{h_1}\hat{c}$ (context drives gates) and $I_c^{(h_1)} = W_{h_1}^\top \hat{s}$ (gates feed back to context).\\
\textbf{(ii) Sparsity-inducing self-synapse ($h_2$).}
The crux of our construction, wherein a positive-energy self-connection $E_{\{h_2\}} = \frac{\alpha}{2}\hat{s}^\top W_{h_2} \hat{s}$ on the gate layer penalizes simultaneous activation of gates associated with similar memories. The weight matrix has entries $W_{h_2}^{ij} = \langle \zeta^i, \zeta^j \rangle$ for $i \neq j$ and zero diagonal, making the penalty proportional to inter-memory similarity and $\alpha\geq0$ is the penalization strength. This produces a competitive signal $I_s^{(h_2)} = -2\alpha W_{h_2}\hat{s}$ that suppresses gates whose memories are close to already-active ones, enforcing sparse gate patterns without explicit regularization. The zero diagonal ensures neurons do not penalized beyond the intrinsic damping from $E^{\text{neuron}}$.\\
\textbf{(iii) Cross-circuit coupling ($h_3$).}
An identity-weighted coupling $E_{h_3} = -\lambda\,\hat{s}^\top \hat{r}$ with scalar strength $\lambda > 0$ bridges the context subcircuit (S--C) and the retrieval subcircuit (R--Q). The bidirectional signals $I_s^{(h_3)} = \lambda\hat{r}$ and $I_r^{(h_3)} = \lambda\hat{s}$ allow context-derived gate activations to bias retrieval while retrieval outcomes refine the gates. Large $\lambda$ forces tight coordination between subcircuits, while small $\lambda$ permits near-independent operation. When the context signal has a larger magnitude, its influence on retrieval dominates the reverse direction.\\
\textbf{(iv) Memory retrieval ($h_4$).}
A bilinear energy $E_{h_4} = -\hat{q}^\top W_{h_4} \hat{r}$ with $W_{h_4} = [\zeta^1, \ldots, \zeta^N] \in \mathbb{R}^{d_q \times N}$ encodes the stored memory patterns. The softmax activations $\hat{r}$ act as attention weights over the columns of $W_{h_4}$, producing a signal $I_q^{(h_4)} = \sum_\mu \zeta^\mu \hat{r}_\mu$ that is a combination of memories. The reverse direction, $I_r^{(h_4)} = W_{h_4}^\top \hat{q}$ computes the similarity of the query and stored memories.

\subsubsection{Evolution Equations}

Combining all synaptic inputs with the self-damping from each neuron layer, we arrive at a succint description for the coupled dynamical system:
\begin{align}
\tau_c \, \dot{c}_i &= (W_{h_1}^\top \hat{s})_i - c_i \label{eq:c_dyn} \\
\tau_s \, \dot{s}_i &= (W_{h_1} \hat{c})_i - \alpha{\textstyle\sum_{j \neq i}} \langle \zeta^i, \zeta^j \rangle\, \hat{s}_j + \lambda\, \hat{r}_i - s_i \label{eq:s_dyn} \\
\tau_r \, \dot{r}_\mu &= \lambda\, \hat{s}_\mu + \langle \zeta^\mu, \hat{q} \rangle - r_\mu \label{eq:r_dyn} \\
\tau_q \, \dot{q}_i &= {\textstyle\sum_\mu} \zeta^\mu_i\, \hat{r}_\mu - q_i \label{eq:q_dyn}
\end{align}

Equation~\eqref{eq:r_dyn} reveals the two-stage mechanism: the pre-softmax score for memory $\mu$ decomposes as a context-driven gate bias $\lambda \hat{s}_\mu$ and a query-memory similarity $\langle \zeta^\mu, \hat{q} \rangle$, with the gate bias able to reshape the retrieval landscape before the query has fully resolved. When the context subcircuit operates at a faster timescale ($\tau_c, \tau_s \ll \tau_r, \tau_q$), the gates effectively precondition retrieval dynamics.

\subsection{Analysis}
\label{sec:theory}

Now, we highlight some theoretical consequences of our architecture.
First, we characterize the sufficient condition for reliable retrieval as a threshold on the separation gap between target and closest competitor. This gap decomposes into raw similarity (as in standard MHNs) plus a context-driven gate contrast, showing how favorable context eases retrieval via $\lambda$. The setup is as follows: we fix a target memory index $\mu$ and query $q$. Define the raw similarity gap
$\Delta_{\mathrm{raw}} := \langle \zeta^\mu, q\rangle - \max_{\nu \neq \mu}\langle \zeta^\nu, q\rangle$,
the gate contrast
$\Delta_{\mathrm{gate}} := s_\mu - \max_{\nu \neq \mu} s_\nu$,
and the effective separation gap
$\Delta := \Delta_{\mathrm{raw}} + \lambda\,\Delta_{\mathrm{gate}}$.

\begin{theorem}\label{thm:separation}
The target pattern $\mu$ is a stable fixed point of the retrieval dynamics with retrieval probability $P(\mu)\ge 1-\epsilon$ if
\[
\Delta = \Delta_{\mathrm{raw}} + \lambda\,\Delta_{\mathrm{gate}} \geq \frac{1}{\beta}\ln\!\left(\frac{(1-\epsilon)(N-1)}{\epsilon}\right).
\]
In particular, when $\Delta_{\mathrm{gate}} > 0$, increasing $\lambda$ increases $\Delta$ and relaxes the required raw separation.
\end{theorem}

 The proof is deferred to Appendix~\ref{app:thm:separation}. Importantly, since context can control $\Delta$ through $\Delta_{\mathrm{gate}}$, including the right context $c$ can directly increase the separation gap between memories through its action on $s$. And because our retrieval component uses a LogSumExp Lagrangian, it takes the form of a standard MHN and inherits its exponentially large storage capacity and exponentially small retrieval error (\cite{ramsauer2020hopfield}) (refer to Appendix~\ref{sec:appendix_MHN} for details on MHNs). Consequently, the increased gap translates into exponential improvements in retrieval, as formalized below.
 
\begin{corollary}[Corollary of Theorems 4 \& 5 in \cite{ramsauer2020hopfield}; Informal]\label{cor:improve_context}
    An increase in the separation $\Delta_i$ between memories gives rise to an exponential improvement in the memory retrieval in one update. Similarly, an increase in $\Delta_i$ exponentially decreases the retrieval error.
\end{corollary}

Next we characterize how the gating mechanism sparsifies selection, even under the absence of any context-based biasing. This will serve to provide insights on our subsequent study on the interaction between context gating and retrieval feedback. To do this, we study the self-synapse $h_2$ in isolation by setting $\lambda=0$ (to exclude query influence) and a \emph{uniform input}, i.e. $u:=W_{h_1} \hat{c}_i = p\ \forall\ i$, so the context does not distinguish between memories. 

Notice then, that the fixed point equation for gate vector $s$ (Eq.~\ref{eq:s_dyn}) is of the form $A(\alpha)s^*=u$, with the linear operator
$
A(\alpha)\triangleq I+\alpha W_{h2}=(1-\alpha)I+\alpha G.
$
Here $G\in\mathbb{R}^{N\times N}$ is the Gram matrix of $N$ unit-norm memories $\{\zeta^\mu\}_{\mu=1}^N\subset\mathbb{R}^d$ with
$G_{ij}=\langle \zeta^i,\zeta^j\rangle$ and $\mathrm{rank}(G)=\min(N,d)=:R$. Such an assumption is easily satisfied with high probability under random sampling of memories. 

To analyze this, we consider the spectral decomposition of $G=\sum_{k=1}^{R}\mu_k\,v_kv_k^\top,\: \mu_1\ge\dots\ge \mu_R>0,$
with $\{v_k\}_{k=1}^R$ orthonormal. When $N>R$ (i.e., $N>d$), we can extend $\{v_k\}_{k=1}^R$ to an orthonormal basis
$\{v_k\}_{k=1}^N$ of $\mathbb{R}^N$ by choosing $\{v_k\}_{k=R+1}^N$ to span $\ker(G)$, with eigenvalue $\mu_k=0$ for $k>R$.
In all cases, we define the smallest eigenvalue $\mu_{\min}$ of $G$ and let $\mathcal{V}_{\min}:=\mathrm{span}\{v_k:\mu_k=\mu_{\min}\}$ denote the corresponding eigenspace, with orthogonal projector
$P_{\min}$. With this setup, we can now characterize our main result on sparsity due to $s(t)$ in softmax based retrieval.

\begin{theorem}[Critical Regime; Informal]\label{thm:wta}
    Let $\alpha_{\mathrm{crit}} \triangleq (1-\mu_{\mathrm{min}})^{-1}$. Let $s(t)$ be the state and $u$ the input. Further, let $p(t) = \mathrm{softmax}(\beta s(t))$ be the probability distribution induced by $s(t)$.

    \textbf{(i) Subcritical} ($\alpha < \alpha_{\mathrm{crit}}$) Then,
    $
    s^* = A(\alpha)^{-1}u = \sum_{k=1}^N \frac{c_k}{\eta_k(\alpha)}v_k, 
    $
    where $c_k$ depends on $v_k$ and the input $u$. Furthermore, the modes $v_k \in \mathcal{V}_{\mathrm{min}}$ are maximally amplified by a factor of $1/\eta_{\min}(\alpha)$ and $\eta_{\mathrm{min}}(\alpha) \to 0$ as $\alpha \to \alpha_{\mathrm{crit}}$.

    \textbf{(ii) Supercritical} ($\alpha > \alpha_{\mathrm{crit}}$) Then,
    $
    p(t) \to \delta_{i^*}(i),\ \text{where} \ i\ \in \{1,2,...,N\}\text{ as } t\to\infty,
    $
    where selected memory $i^*$ is some function of the state and input, along with the spectra of $G$.
\end{theorem}

The proof relies on analyzing the spectra of $A(\alpha)$, which governs the simplified linear dynamics when $\lambda=0$. We leave the details to Appendix~\ref{app:thm:wta}. The key insight is that the selected index $i^*$ is governed by (a) the \emph{geometry} of the memory interactions through the eigenspace $\mathcal{V}_{\min}$ of $G$ corresponding to $\mu_{\min}$, and (b) the \emph{initialization/input bias} w.r.t. the space of stored memories. In what follows, we dive into the interaction between context gating and retrieval by considering $\lambda>0$.

To analyze the dynamics for $\lambda>0$, we focus on the subcritical regime, $A(\alpha)>0$. Because the decoupled gate subsystem ($\lambda=0$) is intrinsically stable in this regime (i.e. converging to a unique finite response $s^*$), it completely isolates the effect of cross-circuit coupling. Any target sharpening or retrieval selectivity that emerges as $\lambda$ increases is therefore strictly driven by the interaction between context gating and retrieval feedback, rather than intrinsic gate instability. The following result formalizes this intuition. It demonstrates that subcritical coupled fixed points are equivalent to those of a self-consistent softmax map, and that this fixed point is unique when the coupling is weak relative to the gate's stability margin.

\begin{theorem}[Self-consistent Retrieval]
\label{thm:self_consistent} Let $u:=W_{h_1}c$ and $b=(b_1,\ldots,b_N)^\top$, where $b_\mu:=\langle \zeta^\mu,q\rangle$.
Assume $A(\alpha)>0$ and define $\eta_{\min}(\alpha):=\lambda_{\min}(A(\alpha)).$ For any $\lambda>0$, every fixed point $(s^*,r^*)$ of the coupled gate-retrieval subsystem satisfies,
\begin{equation}
    s^*=A(\alpha)^{-1}(u+\lambda p^*),\quad r^*=b+ \lambda A(\alpha)^{-1} u + \lambda ^2 A(\alpha)^{-1} p^*,
    \label{eq:comb_fix_point}
\end{equation}
where $p^*=\mathrm{softmax}_\beta(r^*)$. Equivalently, $p^*$ is a fixed point of the self-consistent map,
\begin{equation}\label{eq:self-consistent-map}
    p^*=\Phi_{\alpha,\lambda}(p^*):= \mathrm{softmax}_\beta \left(b+ \lambda A(\alpha)^{-1} u+\lambda^2 A(\alpha)^{-1} p^* \right)
\end{equation}
Conversely, any fixed point $p^*=\Phi_{\alpha,\lambda}(p^*)$ induces a fixed point $(s^*,r^*)$ through \eqref{eq:comb_fix_point}. Furthermore, if $\frac{\beta\lambda^2}{2\eta_{\min}(\alpha)}<1$, the subystem has a unique fixed point.
\end{theorem}
The proof relies on manipulating equations \ref{eq:s_dyn} and \ref{eq:r_dyn} in their equilibrium states and on showing that $\Phi_{\alpha,\lambda}$ is a contraction mapping. We leave the details of the proof to Appendix~\ref{app:self_consistent}. Theorem \ref{thm:self_consistent} reveals that, in the subcritical regime, positive coupling reshapes retrieval through the effective logit decomposition:
\begin{equation*}
        r^*
    =
    \underbrace{b}_{\text{query evidence}}
    +
    \underbrace{\lambda A(\alpha)^{-1} u}_{\text{context-filtered bias}}
    +
    \underbrace{\lambda^2 A(\alpha)^{-1} p^*}_{\text{retrieval-gate feedback}}.
\end{equation*}
The first term is the usual query-memory similarity. The second term shows how context biases retrieval after being filtered through the stable gate operator $A(\alpha)^{-1}$. The third term is a self-consistent feedback term: retrieval probabilities feed back into the gates, and the gates in turn bias retrieval. Thus, $\lambda$ has two distinct effects: a first-order direct context effect and a second-order feedback effect.

The condition $\frac{\beta\lambda^2}{2\eta_{\min}(\alpha)}<1$ also has a natural interpretation. The quantity $\eta_{\min}(\alpha)$ measures the stability margin of the gate subsystem. As $\alpha\to\alpha_{\mathrm{crit}}$, this margin shrinks, so even small values of $\lambda$ can substantially affect retrieval. Thus, context coupling becomes more influential near the gate critical regime. At the same time, the theorem shows that when $\lambda$ is sufficiently weak relative to this stability margin, the coupled system admits a unique self-consistent retrieval state. This separates two phenomena: below the contraction threshold, context smoothly biases retrieval, whereas beyond it, retrieval-gate feedback can bias selection or winner-take-all behavior.

\subsection{Simulations}\label{sec:experiments}
We validate the properties of our framework through synthetic experiments. First, we look at how context coupling affects retrieval, by measuring how varying $\lambda$ impacts $\Delta$, and track the resulting changes in retrieval accuracy and converged probabilities. Then, we study the sparsity-inducing mechanism in isolation by empirically verifying the phase transition characterized in Theorem~\ref{thm:wta}. Due to space constraints, we defer the exact experimental setups to the Appendix.

\textbf{Context-Augmented Memory Separation.} Figure~\ref{fig:separation}(a) shows retrieval accuracy as a function of query noise. The accuracy is non-decreasing in $\lambda$, with the largest gains in the intermediate noise regime. Context does not eliminate error at extreme noise, but it substantially extends the operational noise range. Figure~\ref{fig:separation}(b) plots retrieval probability against $\Delta$ for all trials. All empirical points lie on or above the exact logistic lower bound (red curve), confirming Theorem~\ref{thm:separation}. The primary effect of increasing $\lambda$ is to shift the point cloud rightward along the $\Delta$-axis toward larger gaps. Figure~\ref{fig:separation}(c) makes this effect explicit: the box plots show the distribution of $\Delta$ at fixed query noise $\sigma_q=1.0$ for each $\lambda$. As $\lambda$ increases, the entire gap distribution shifts upward, with higher $\lambda$ values pushing the median gap well above the exact threshold required.

\begin{figure}
    \centering
    \includegraphics[width=0.8\linewidth]{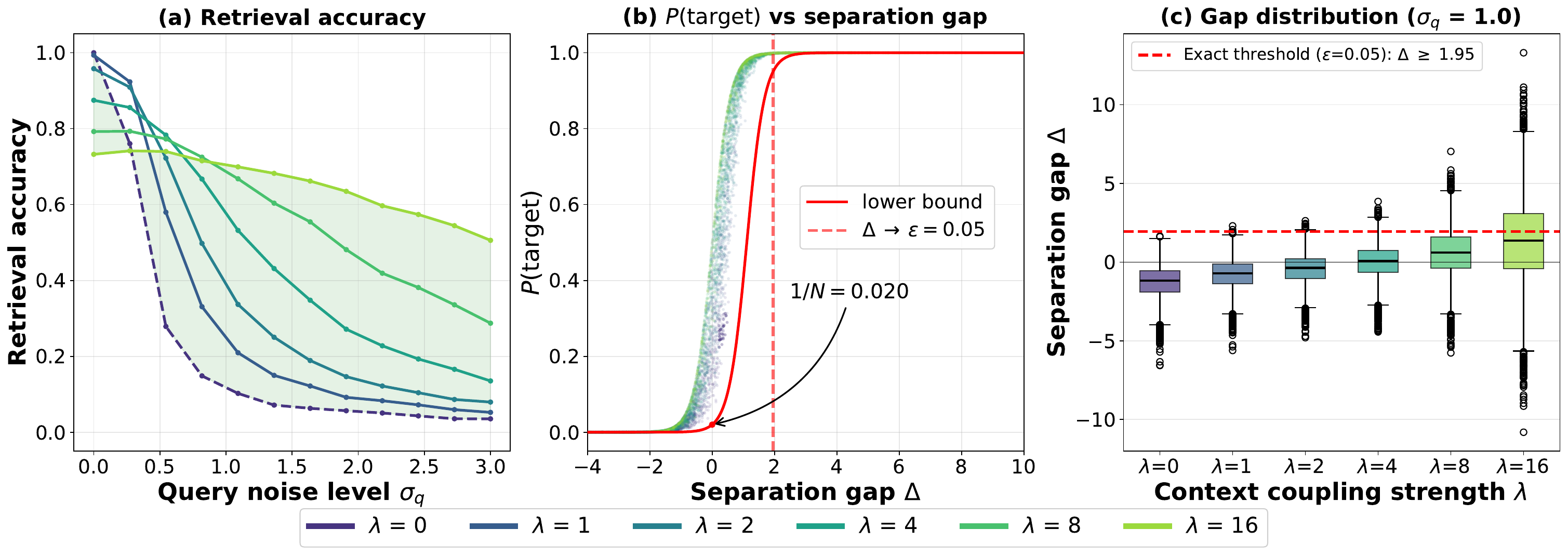}
    \caption{Context-augmented memory separation. \textbf{(a)}~Retrieval accuracy vs. query noise. \textbf{(b)}~Retrieval probability vs. effective separation gap $\Delta$.  
    \textbf{(c)}~Distribution of $\Delta$ at $\sigma_q=1.0$ for each $\lambda$.
    }
    \label{fig:separation}
\end{figure}

\textbf{Phase Transition Experiments} Figure~\ref{fig:competition}(a) shows that the peak gate probability $\max_i p^*_i$ undergoes a sharp transition at $\alpha_{\mathrm{crit}}$: for $\alpha < \alpha_{\mathrm{crit}}$, it stays near $1/N$, consistent with the spectral filtering regime where $1/\eta_N(\alpha)$ remains finite; for $\alpha > \alpha_{\mathrm{crit}}$, it saturates near $1.0$, matching the WTA limit of Theorem~\ref{thm:wta}. Figure~\ref{fig:competition}(b) resolves the transition region $[\alpha_{\mathrm{crit}},\, 1.18]$, revealing a smooth but steep curve. An example gate distribution at $\alpha = \alpha_{\mathrm{crit}}$ (Figure~\ref{fig:competition}(c)) shows the onset of selectivity: multiple memories emerge above the uniform baseline, reflecting the divergent amplification of the least-correlated mode $v_N$ as $\eta_N \to 0^+$. Figure~\ref{fig:competition}(d) showcases the transition of the critical region as $\lambda$ increases. As $\lambda$ increases the empirical $\alpha_\mathrm{crit}$ shifts to smaller values. While this agrees with our preliminary results in Appendix \ref{app:thm:lambda}, we emphasize that our theoretical threshold does not match the empirical threshold. Nevertheless, this suggests that context gating indeed aids retrieval.

\begin{figure}[ht]
    \centering
    \includegraphics[width=\textwidth]{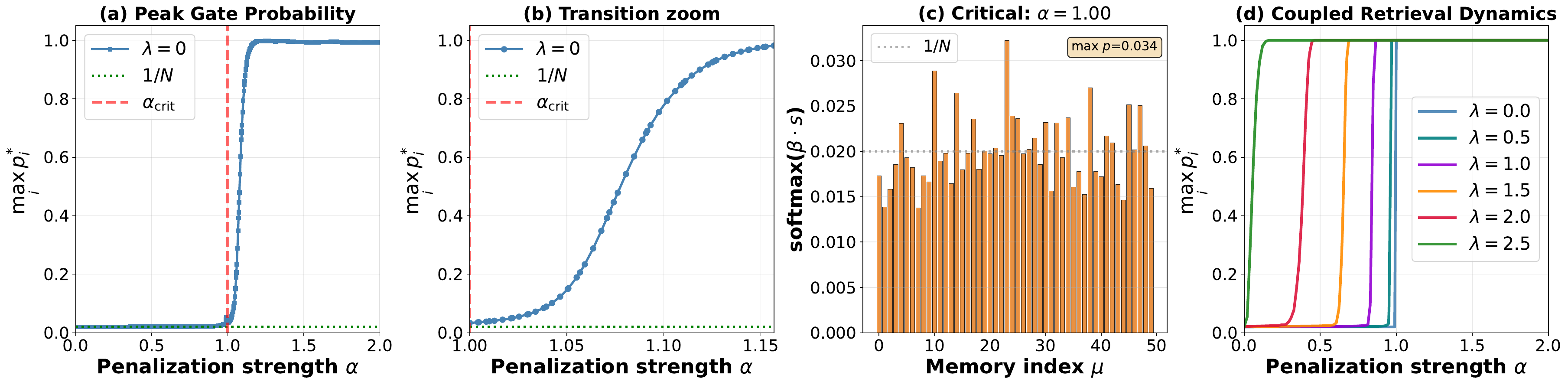}
    \caption{Phase transition in gate selectivity.
    \textbf{(a)} Peak gate probability vs. penalization strength ($\lambda = 0$).
    \textbf{(b)} Zoom on the transition region ($\lambda = 0$).
    \textbf{(c)} Example gate distribution at $\alpha = \alpha_{\mathrm{crit}}$ ($\lambda = 0$). \textbf{(d)} Peak gate probability vs. penalization strength for varying $\lambda$.}\label{fig:competition}
\end{figure}

\section{Connection to Transformers}

Having formalized context-dependent retrieval in our associative
memory, we question the extent to which this structure is reflected in the internal computations of the LLMs that inspired it. \cite{yangCZI2025} shows that hidden-state geometry in transformer ICL exhibits a two-phase structure: \emph{separability} (class-level clustering) emerging in early-to-mid layers, \emph{alignment} (label-specific directionality) sharpening
later, paralleling the timescale separation in our architecture, where the context circuit ($\tau_c, \tau_s \ll \tau_r, \tau_q$) settles before retrieval.\\
Note that embedding a full LLM within our framework is not feasible, given that our abstraction operates on limited-capacity representations as opposed to billions of parameters.%
\footnote{Scaling associative memories to match, e.g., via distributed representations with exponential capacity~(\cite{kafrajKL2026}), while possible, is beyond our scope.} Instead, we assume access to pretrained representations and propose to analyze the retrieval component in isolation. This is grounded in the linear representation hypothesis~(\cite{parkCV2024}) and work on in-context~(\cite{liuYXZ2024}) and task vectors~(\cite{dongJZN2025}), which
establish that high-level concepts are encoded as directions in
representation space and that demonstration effects can be distilled into additive vectors, thus supporting our decomposition of hidden states into a query term plus an additive context signal.

\subsection{Setup}

To instantiate our framework within a pretrained transformer, we first identify each architectural component with an extractable quantity. \\
\textbf{Memory bank}: A natural choice for the stored patterns 
$\{\zeta^\mu\}_{\mu=1}^N$ is the rows of the unembedding matrix 
$W_U \in \mathbb{R}^{N \times d}$, since the output logit for token $\mu$ is 
the inner product $\langle \zeta^\mu, h \rangle$, directly mirroring the retrieval 
score in our architecture.\\
\textbf{Context-gate projection:} Recall that the gate activation for memory 
$\mu$ is $s_\mu = (W_{h_1} c)_\mu$. Writing $W_{h_1} = Z^\top P$ where 
$Z = W_U$ and $P \in \mathbb{R}^{d \times d_c}$ is a learned projection, we 
obtain $s_\mu = \langle \zeta^\mu, Pc \rangle$: the gate scores each memory by 
its similarity to the projected context. As a first-order simplification, we 
set $P = I$, so that gate activations reduce to direct inner products between 
context and memory vectors. Learning $P \neq I$ is an important direction that we defer to future work.\\
\textbf{Context and query extraction:} For each query $i$, we extract hidden states with demonstrations ($h_{\mathrm{ICL},i}^{(\ell)}$) and without ($h_{\mathrm{zero},i}^{(\ell)}$) at each layer $\ell$. The query is 
$q_i = h_{\mathrm{zero},i}^{(\ell)}$. For the context, following 
function vector methodology of~\cite{toddLSMWB2024}, we average the additive effect of demonstrations across queries:
$$
\bar{c}^{(\ell)} = \frac{1}{n}\sum_{i=1}^{n} 
\bigl(h_{\mathrm{ICL},i}^{(\ell)} - h_{\mathrm{zero},i}^{(\ell)}\bigr),
$$
which removes any query-specific information to yield a shared task-level signal. The 
retrieval score for memory $\mu$ is then simply
\begin{equation}
    r_\mu = \langle \zeta^\mu, q_i \rangle + \lambda\, \langle \zeta^\mu, \bar{c}^{(\ell)} \rangle
    \label{eq:empirical_retrieval}
\end{equation}
mirroring \eqref{eq:r_dyn}, where $\lambda$ controls the coupling between context and retrieval.

In terms of the task, our testbed is in-context learning: a set of demonstrations (the context)
precedes a query, and the model must predict the correct target token. We use Llama-3-8B (32 layers) on SST-2 ($K{=}2$) \cite{socherPWCMNP2013}, AG-News
($K{=}4$) \cite{zhangZL2016}, and TREC ($K{=}6$) \cite{liT2002} for
classification, where $L$ denotes the label-set size; and on LAMA
\cite{petroniRLBWMR2019} for open-ended factual retrieval. For brevity, we focus the main text on the binary ICL classification task on SST-2; results on AG-News and TREC follow similar trends and are reported in App.~\ref{app:add_datasets}. Finally, we sweep the coupling $\lambda$ throughout.

\textbf{Remark on Empirical Simplifications:} We adapt our continuous-time equilibrium model (Theorem~\ref{thm:self_consistent}) to feed-forward transformers via two structural simplifications that lead to Eq.~\ref{eq:empirical_retrieval}. First, we set the sparsity penalty $\alpha=0$, which reduces the gate operator $A(\alpha)$ to the identity matrix. While our theory relies on $\alpha$ to regularize the memory Gram matrix $G = W_U W_U^\top$, the heavily skewed spectrum of actual LLM unembedding matrices makes the stable $\alpha$ regime practically degenerate. Furthermore, native transformer dynamics directly induce this geometric sparsification without explicit regularization (as seen via the $N_{\mathrm{eff}}$ collapse in Fig.~\ref{fig:neff_layers}). Second, because transformers lack single-layer recurrence, we drop the second-order feedback term ($\lambda^2$). Together, these reductions naturally yield our linear empirical score, isolating the immediate, first-order contextual bias ($\lambda$) while leaving complex non-linear feedback to be handled implicitly.

\subsection{Experiments}

(i) We first examine how the transformer itself processes demonstrations,
without imposing our additive decomposition. We decode the hidden states $h_{\mathrm{zero}}^{(\ell)}$ and $h_{\mathrm{ICL}}^{(\ell)}$ directly through the unembedding at each layer to obtain distributions over memories $\{\zeta^\mu\}$, tracking how the effective number of active memories $N_\mathrm{eff}$ evolves across layers and shot counts. The contrast between the two reveals the layers at which native ICL processing collapses the memory space onto the label set, providing a guidance on where to extract informative context representations. \\
\textbf{Takeaway}: Native transformer dynamics implicitly execute the geometric sparsification and timescale separation predicted by our architecture (Fig.~\ref{fig:neff_layers}). The zero-shot query remains broadly distributed, whereas the introduction of context induces a sharp $N_{\mathrm{eff}}$ collapse (onto the label set) at later layers. This demonstrates that transformers natively resolve the fast-context/slow-retrieval separation, functionally replacing our theoretical requirement for an explicit sparsity penalty ($\alpha$), and aligns with recent findings on representational geometry \cite{yangCZI2025} (see App.~\ref{app:replicate_icl_geom} for the replication).

\begin{figure}[ht]
\centering
\includegraphics[width=0.8\linewidth]{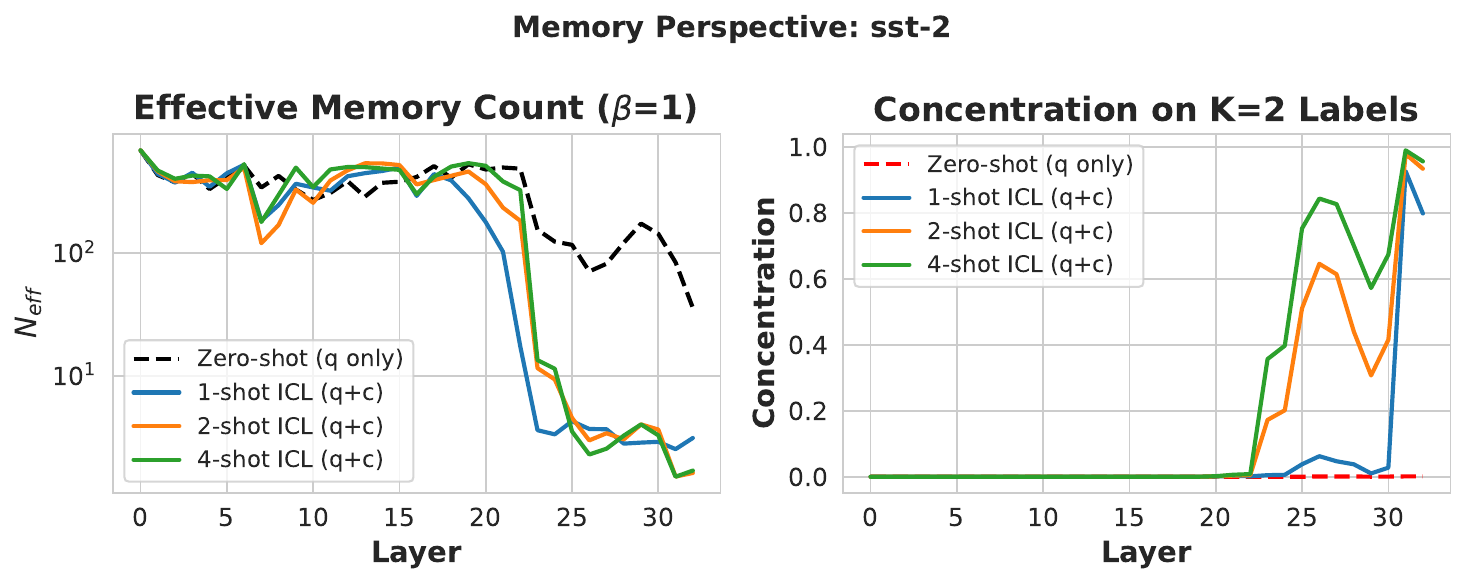}
\caption{\textbf{Native ICL processing collapses the memory space onto the
label set.} Effective number of active memories $N_\mathrm{eff}$ obtained
by decoding $h_{\mathrm{zero}}^{(\ell)}$ and $h_{\mathrm{ICL}}^{(\ell)}$
through the unembedding at each layer, shown across layers and shot
counts on SST-2.}
\label{fig:neff_layers}
\end{figure}

(ii) We next evaluate our additive retrieval score by extracting averaged context $\bar{c}^{(\ell)}$ from different layers and pairing it with the zero-shot query $h_{\mathrm{zero}}^{(\ell')}$, sweeping $\lambda$ to
identify which (context-layer, query-layer) combinations yield the
strongest retrieval and how this scales with shots. \\
\textbf{Takeaway}: Fig~\ref{fig:retrieval_sweep} is arguably our strongest result, showing how additive energy landscape reshaping drives in-context learning. Despite the first-order approximation of our fixed-point theory, the simple linear combination of an isolated context vector and a zero-shot query achieves up to $\sim$86\% accuracy, successfully capturing the vast majority of the predictive power of full ICL ($\sim$95\%). This confirms that context primarily acts as a structural prior that primes the retrieval space, with more informative contexts yielding sharper retrieval. 
 
\begin{figure}[ht]
\centering
\begin{subfigure}{\linewidth}
\centering
\includegraphics[width=0.8\linewidth]{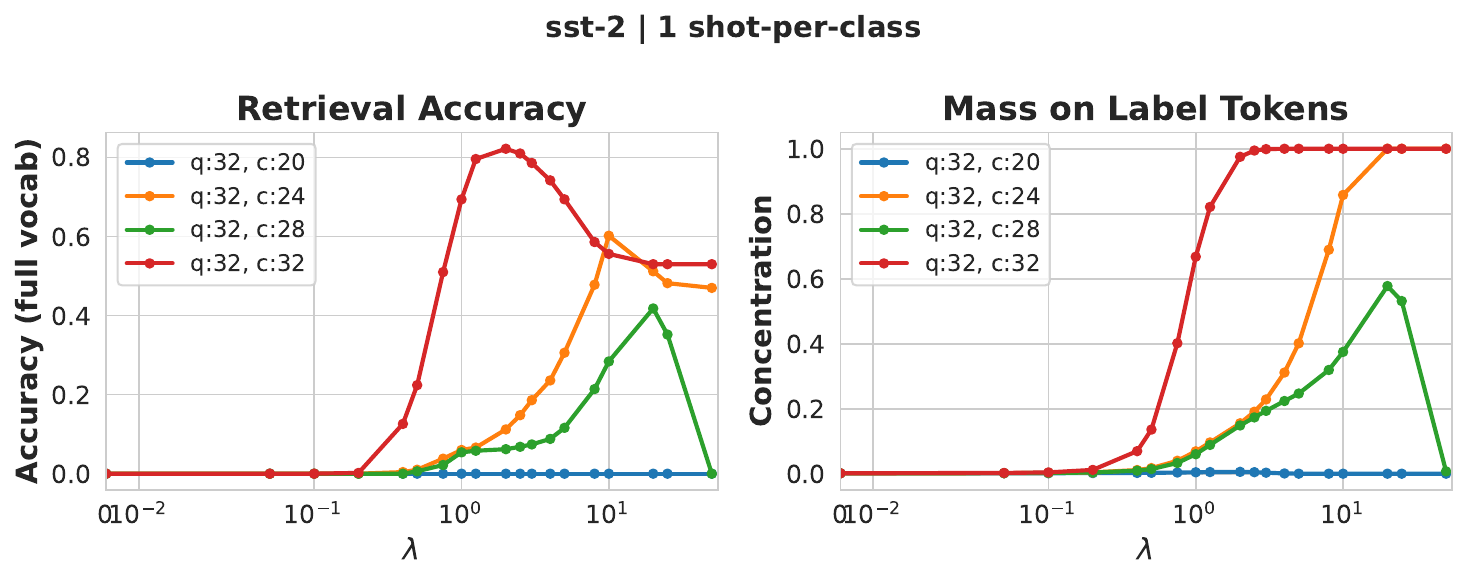}
\label{fig:retrieval_a}
\end{subfigure}
\begin{subfigure}{\linewidth}
\centering
\includegraphics[width=0.8\linewidth]{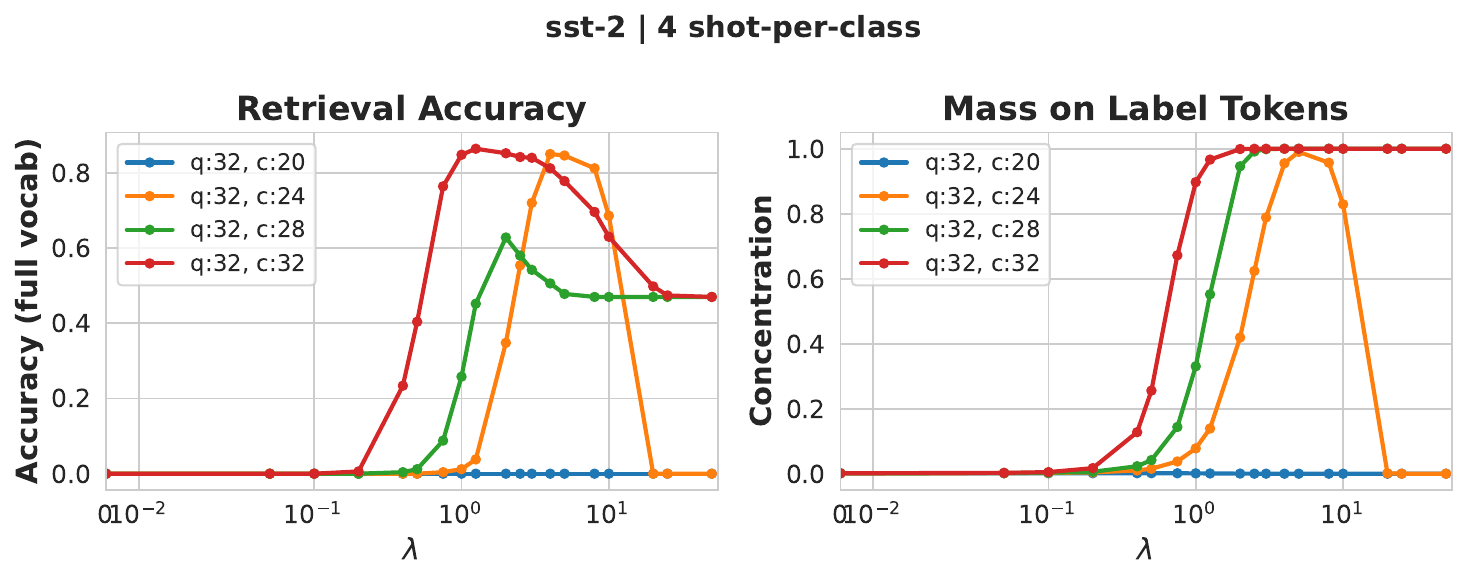}
\label{fig:retrieval_b}
\end{subfigure}
\caption{\textbf{Coupled retrieval across context-query layer
combinations.} Retrieval accuracy as a function of $\lambda$, context
extraction layer $\ell$, and number of shots (1 and 4) on SST-2.}
\label{fig:retrieval_sweep}
\end{figure}

(iii) Finally, we extend the analysis to open-ended factual retrieval. Each query is a fact-completion prompt of the form ``the capital
of France is''; we extract its zero-shot representation at the final
layer ($\ell{=}32$) and construct two contrasting context signals from
$k$-shot demonstrations: a positive $\bar{c}_{+}^{(\ell)}$ built from
correct subject-object pairs, and a negative $\bar{c}_{-}^{(\ell)}$ built
from deliberately wrong ones. Holding the query fixed and sweeping the
context layer over $\ell \in \{0, \ldots, 32\}$ at three values of
$\lambda$, we compare how supportive and adversarial context shape
retrieval over the memory bank. \\
\textbf{Takeaway}: Additive context bidirectionally steers the retrieval energy landscape. Consistent with our associative memory framework and prior work on task steering vectors \cite{liuYXZ2024, dongJZN2025}, injecting supportive context deepens the target's energy minima (sharpening retrieval), while adversarial context actively shifts the distribution mass toward competing attractors (Fig.~\ref{fig:open_fr}). This provides an interesting associative memory perspective on the empirical success of task steering vectors.

\begin{figure}[ht]
\centering
\includegraphics[width=0.8\linewidth]{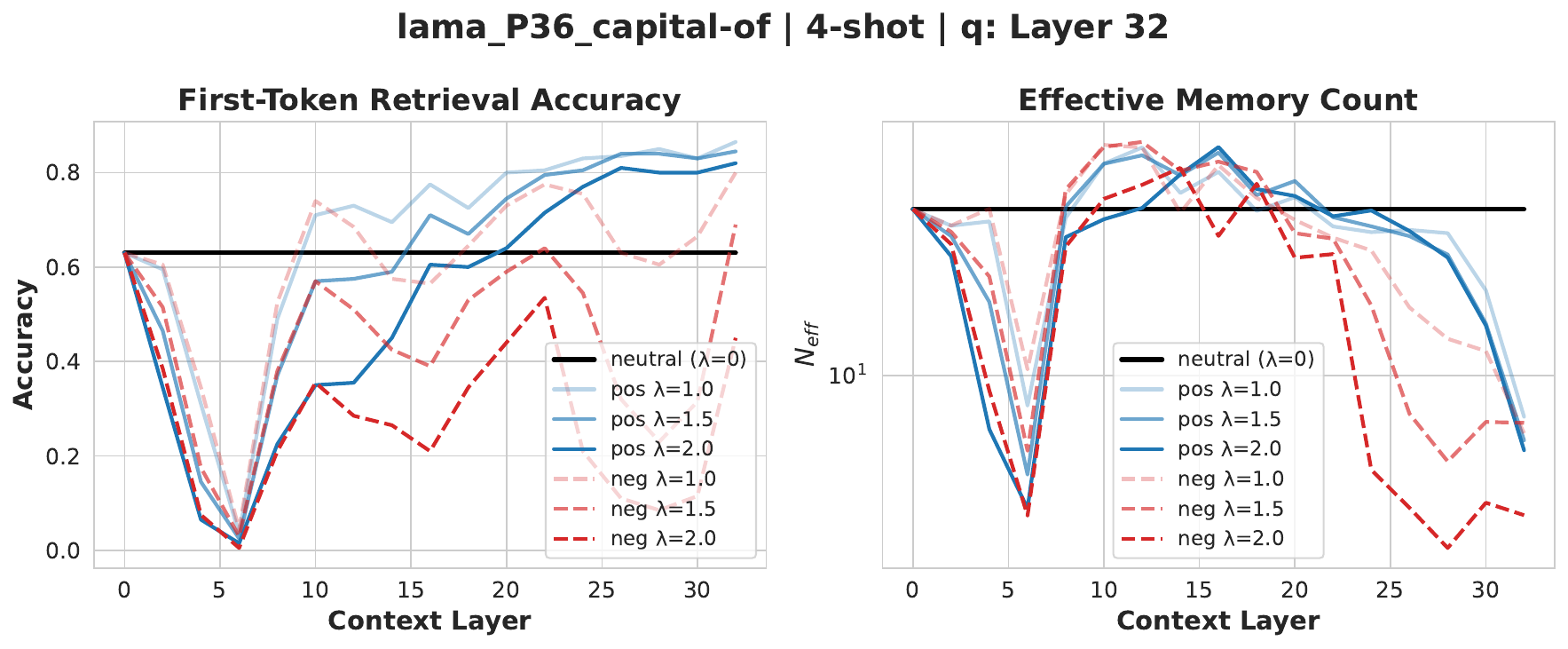}
\caption{\textbf{Additive context bidirectionally steers factual
retrieval.} Retrieval over the memory bank for a LAMA query held at
layer $\ell{=}32$, paired with positive ($\bar{c}_{+}^{(\ell)}$, correct demonstrations) and negative ($\bar{c}_{-}^{(\ell)}$, wrong
demonstrations) context signals extracted across $\ell \in \{0, \ldots,
32\}$ at three values of $\lambda$.}
\label{fig:open_fr}
\end{figure}

\section{Discussion}
\label{sec:discussion}

In this work, we introduced a two-stage associative memory architecture in which a context-gate subcircuit dynamically reshapes the retrieval energy landscape before and during recall. Our mathematical analysis establishes three core properties of this framework. First, context coupling directly increases the effective separation gap between memories, which translates into exponential gains in retrieval accuracy. Second, the isolated gate subsystem features a critical phase transition that natively induces winner-take-all sparsity. Finally, under subcritical coupling, the bidirectional system admits a unique self-consistent fixed point. This fixed point cleanly decomposes the retrieval state into raw query evidence, a direct context-filtered bias, and a non-linear retrieval-gate feedback loop, isolating the distinct mechanisms of contextual influence. Beyond theory, we bridged this framework to the internal representations of LLMs. By evaluating a first-order approximation of our model, we demonstrated that the native geometric dynamics of transformers functionally mirror our proposed architecture, utilizing context to localize the searchable memory subspace before the query discriminates within it.

\textbf{Limitations and Future Work.}
While our theory guarantees the existence of a unique self-consistent fixed point, our current analysis does not fully characterize the continuous-time trajectory of the system, meaning we lack strict mathematical guarantees of convergence from arbitrary initializations. A highly promising theoretical direction is to pose the joint context-retrieval optimization as an entropy maximization problem with score constraints. From this perspective, the coupled dynamics could be analyzed as a mirror descent algorithm, potentially yielding rigorous convergence guarantees via Lyapunov stability and drawing deep connections between our context gating and deterministic annealing. 

Empirically, our transformer evaluation isolates the first-order contextual bias but omits the second-order feedback loop ($\lambda^2$), restricting our testbed to single-token outputs in classification and simple factual recall. Future work must extend this framework to multi-token autoregressive generation and highly complex, open-ended contexts. For instance, analyzing how behavioral system prompts or persona adoption (e.g., ``act as a liar") continuously reshape the energy landscape over long horizons would provide a rigorous associative memory foundation for complex LLM phenomenology. Finally, learning a dedicated context-gate projection matrix ($P \neq I$) rather than relying on direct inner products remains a vital next step to fully operationalize our architecture in pre-trained models.

\section*{Acknowledgements}

We are thankful to Dmitry Krotov for early discussions on modeling context-based retrieval.

\bibliography{nfam2026_workshop}
\bibliographystyle{plain}

\newpage
\appendix

\section{Primer: General Dense Associative Memories}
\label{sec:appendix_primer_gdam}

\subsection{Modular Energy Framework}

Here, we provide a brief introduction on the generalized abstraction of Energy-based AMs (HAMUX), introduced by \cite{krotovHPP2025}. At a high level, dense associative memories (DAMs) are composed of two major components: \textbf{neuron layers} (nodes) and \textbf{hypersynapses} (hyperedges). A neuron layer captures a non-linearity in the network (or activations such as ReLU), while a hypersynapse is a parameterized energy function that captures how similar or aligned the activations of its connected neuron layers are. The DAM is then described by summing the modular energies of all these components. The following exposition provides a gentle introduction. For a formal treatment, we point the reader to ~\cite{krotov2021ham, krotovHPP2025}.

\subsubsection{Total Energy}

For a system with $L$ neuron layers and $S$ hypersynapses, the total energy is:
\begin{equation}
E_{\text{total}} = \sum_{\ell=1}^{L} E_{\ell}^{\text{neuron}} + \sum_{s=1}^{S} E_{s}^{\text{synapse}}
\end{equation}

\subsubsection{Local Update Rule}

Let $\hat{\mathbf{x}}_\ell$ and $\mathbf{x}_\ell$ denote the activations and internal states of neuron layer $\ell$. Let $\mathcal{N}(\ell)$ be the set of hypersynapses connected to layer $\ell$. The internal states minimize total energy via:
\begin{equation}
\tau_\ell \frac{d\mathbf{x}_\ell}{dt} = -\frac{\partial E_{\text{total}}}{\partial \hat{\mathbf{x}}_\ell} = -\left(\sum_{s \in \mathcal{N}(\ell)} \frac{\partial E_s^{\text{synapse}}}{\partial \hat{\mathbf{x}}_\ell}\right) - \frac{\partial E_\ell^{\text{neuron}}}{\partial \hat{\mathbf{x}}_\ell} = \mathbf{I}_{\mathbf{x}_\ell} - \mathbf{x}_\ell
\end{equation}
where $\mathbf{I}_{\mathbf{x}_\ell} := -\sum_{s \in \mathcal{N}(\ell)} \nabla_{\hat{\mathbf{x}}_\ell} E_s^{\text{synapse}}$ is the total synaptic input current and $\tau_\ell$ is the time constant.

\subsection{Dynamical Neurons and Lagrangians}

\begin{definition}[Neuron Layer]
A neuron layer has internal state $\mathbf{x}$ and activation $\hat{\mathbf{x}}$, defined via a convex Lagrangian $\mathcal{L}_\mathbf{x}(\mathbf{x})$ and its Legendre transform $\mathcal{T}$:
\begin{align}
\hat{\mathbf{x}} &= \nabla \mathcal{L}_\mathbf{x}(\mathbf{x}) \quad \text{(activation function)} \\
E_\mathbf{x}(\hat{\mathbf{x}}) &= \mathcal{T}[\mathcal{L}_\mathbf{x}] = \langle \mathbf{x}, \hat{\mathbf{x}} \rangle - \mathcal{L}_\mathbf{x}(\mathbf{x}) \quad \text{(dual energy)}
\end{align}
where $\langle \cdot, \cdot \rangle$ is the element-wise inner product.
\end{definition}

\textbf{Key Property:} The energy gradient equals the hidden state:
\begin{equation}
\frac{\partial E_\mathbf{x}(\hat{\mathbf{x}})}{\partial \hat{\mathbf{x}}} = \mathbf{x}
\end{equation}

\textbf{Proof:}
\begin{align*}
\frac{\partial E_\mathbf{x}(\hat{\mathbf{x}})}{\partial \hat{\mathbf{x}}} &= \frac{\partial}{\partial \hat{\mathbf{x}}} \left(\langle \mathbf{x}, \hat{\mathbf{x}} \rangle - \mathcal{L}_\mathbf{x}(\mathbf{x})\right) \\
&= \mathbf{x} + \hat{\mathbf{x}} \frac{\partial \mathbf{x}}{\partial \hat{\mathbf{x}}} - \frac{\partial \mathcal{L}_\mathbf{x}(\mathbf{x})}{\partial \mathbf{x}} \frac{\partial \mathbf{x}}{\partial \hat{\mathbf{x}}} \\
&= \mathbf{x} + \hat{\mathbf{x}} \frac{\partial \mathbf{x}}{\partial \hat{\mathbf{x}}} - \hat{\mathbf{x}} \frac{\partial \mathbf{x}}{\partial \hat{\mathbf{x}}} = \mathbf{x}
\end{align*}

This implies exponential decay in isolation:
\begin{equation}
\frac{d\mathbf{x}}{dt} = -\nabla_{\hat{\mathbf{x}}} E_\mathbf{x}(\hat{\mathbf{x}}) = -\mathbf{x}
\end{equation}

\subsection{Hypersynapses}

\begin{definition}[Hypersynapse]
A hypersynapse is a scalar-valued energy function defined on the activations of connected neuron layers. For layers $X$ and $Y$ with activations $\hat{\mathbf{x}}$ and $\hat{\mathbf{y}}$:
\begin{equation}
E_{xy}(\hat{\mathbf{x}}, \hat{\mathbf{y}}; \Xi)
\end{equation}
where $\Xi$ represents learnable synaptic weights. Low energy indicates activations satisfy the relationship encoded by $\Xi$.
\end{definition}

\textbf{Example:} Dense hypersynapse: $E_{\text{Dense}}(\hat{\mathbf{x}}, \hat{\mathbf{y}}; \Xi) = -\hat{\mathbf{x}}^\top \Xi \hat{\mathbf{y}}$

\textbf{Key Properties:}
\begin{itemize}
\item Hypersynapses connect any number of layers (hyperedges)
\item Undirected: all connected layers influence each other bidirectionally
\item Signal to layer $X$: $\mathbf{I}_x = -\nabla_{\hat{\mathbf{x}}} E_{xy}(\hat{\mathbf{x}}, \hat{\mathbf{y}}; \Xi)$
\item Signal to layer $Y$: $\mathbf{I}_y = -\nabla_{\hat{\mathbf{y}}} E_{xy}(\hat{\mathbf{x}}, \hat{\mathbf{y}}; \Xi)$
\end{itemize}

\textbf{Hypersynapse Notation Conventions:}

For synapses connecting multiple layers, we subscript with the identifiers of all connected layers:
\begin{itemize}
\item $E_{xy}$ — synapse connecting layers $X$ and $Y$
\item $E_{xyz}$ — synapse connecting layers $X$, $Y$, and $Z$
\item $E_{xyz\ldots}$ — synapses connecting more than three layers (possible but rare)
\end{itemize}

\textbf{Self-connections:} To avoid confusion with neuron layer energy $E_\mathbf{x}$, we use curly brackets for synaptic self-connections:
\begin{itemize}
\item $E_{\{x\}}$ — interaction energy of a synapse connecting layer $X$ to itself
\end{itemize}

Note: Since almost every interaction energy is parameterized, we generally omit $\Xi$ from notation when not central to the discussion.

\subsection{Energy Descent Dynamics}

\begin{theorem}[Guaranteed Energy Descent]
The dynamics in Eq. (2) decrease the global energy:
\begin{equation}
\frac{dE_{\text{total}}}{dt} = \sum_{\ell=1}^{L} \frac{\partial E_{\text{total}}}{\partial \hat{\mathbf{x}}_\ell} \frac{\partial \hat{\mathbf{x}}_\ell}{\partial \mathbf{x}_\ell} \frac{d\mathbf{x}_\ell}{dt} = -\sum_{\ell=1}^{L} \tau_\ell \frac{d\mathbf{x}_\ell}{dt} \frac{\partial^2 \mathcal{L}_\mathbf{x}}{\partial \mathbf{x}_\ell \partial \mathbf{x}_\ell} \frac{d\mathbf{x}_\ell}{dt} \leq 0
\end{equation}
\end{theorem}

The Hessian $\frac{\partial^2 \mathcal{L}_\mathbf{x}}{\partial \mathbf{x}_\ell \partial \mathbf{x}_\ell}$ is positive semi-definite due to convexity of $\mathcal{L}_\mathbf{x}$, ensuring energy never increases.

\textbf{Convergence:} If energy is bounded below:
\begin{itemize}
\item Strictly positive definite Hessian $\Rightarrow$ converges to fixed point
\item Zero eigenvalues $\Rightarrow$ may converge to fixed manifold with non-zero velocity
\end{itemize}
\subsection{Summary}

To design an Associative Memory:
\begin{enumerate}
\item Choose convex Lagrangian $\mathcal{L}_\mathbf{x}(\mathbf{x})$ for each neuron layer (defines activation function)
\item Design hypersynapse energies $E_s^{\text{synapse}}$ encoding desired relationships
\item Total energy: sum of all neuron and hypersynapse energies
\item Dynamics: minimize energy via local gradient descent (Eq. 2)
\item Guaranteed convergence with bounded activations
\end{enumerate}

\section{Modern Hopfield Networks}\label{sec:appendix_MHN}
To keep this work self-contained, we provide a brief introduction to Modern Hopfield Networks (MHN) (\cite{ramsauer2020hopfield}). MHNs are a form of DAM with an energy function of the form
\begin{equation}
E = -\text{lse}(\beta, X^T\xi) + \frac{1}{2}\xi^T\xi + \beta^{-1}\log N + \frac{1}{2}M^2
\end{equation}
where $\text{lse}(\beta, \cdot)$ is the LogSumExp function with temperature parameter $\beta$, $N$ is the number of stored patterns, $M$ is the norm of the largest pattern (i.e. $M = \max_i \|x_i\|$), and $\xi$ is the query. The update rule for the MHN is then given by
\begin{equation}\label{mhn_update}
\xi^{new} = X\text{softmax}(\beta X^T \xi).
\end{equation}

\cite{ramsauer2020hopfield} provide theorems to guarantee convergence for this form of Associative Memory and further establishes an exponential storage capacity (see Theorems 1, 2, and 3 in their paper). Of particular interest in this paper are Theorems 4 and 5.
\begin{theorem}[Theorem 4 in \cite{ramsauer2020hopfield}]
    With query $\xi$, pattern $x_i$, fixed point $x_i^*$, and separation of $x_i$ to other memories $\Delta_i$, after one update, the distance between $f(\xi)$ and $x_i^*$ is exponentially small. Specifically, 
    \[
        \|f(\xi) - x_i^*\| \leq 2\beta NM^2(N-1)\exp(-\beta(\Delta_i-2\max\{\|\xi - x_i\|, \|x_i^*-x_i\|M))\|\xi - x_i^*\|.
    \]
\end{theorem}

\begin{theorem}[Theorem 5 in \cite{ramsauer2020hopfield}]
    The retrieval error $\|f(\xi) - x_i\|$ is bounded by
    \[
    \|f(\xi) - x_i\| \leq 2(N-1)\exp(-\beta(\Delta_i-2\max\{\|\xi - x_i\|, \|x_i^*-x_i\|M))M
    \]
    and for $\|x_i - x_i^*\| \leq \frac{1}{2\beta M}$ together with $\|x_i - \xi\| \leq \frac{1}{2\beta M}$ we have
    \[
    \|x_i - x_i^*\| \leq 2e(N-1)M\exp(-\beta\Delta_i).
    \]
\end{theorem}

We see from these theorems that increasing the separation between memories yields an exponential improvement in both retrieval after one update and retrieval error.

\subsection{MHN's connection to Transformers}
Notably, MHNs bear a strong mathematical resemblance to the attention mechanism in Transformers (\cite{vaswani2017attention}). This connection can be seen below.

Suppose we have $N$ stored patterns $y_i$ and $S$ state query patterns $r_i$ that are mapped to a space of dimension $d_k$. Set $x_i = W_K^Ty_i$ and $\xi_i = W_Q^t r_i$ and then multiply the result of the update rule (Eq. \ref{mhn_update}) by $W_V$ where $W_K \in \mathbb{R}^{d_y\times d_k}$, $W_Q \in \mathbb{R}^{d_r \times d_k}$, $W_V \in \mathbb{R}^{d_k \times d_v}$.

If we combine everything into matrix operations as is commonly done for attention, we arrive at the following. Let $Y = (y_1, \ldots, y_N)^T$, $R = (r_1, \ldots, r_N)^T$. Define $X^T = K = YW_K$, $\Xi^T = Q = RW_Q$, and $V = YW_KW_V = X^TW_V$. Let the temperature parameter in MHN $\beta = \frac{1}{\sqrt{d_k}}$ and let the output of softmax be a row-vector. Then for the update rule in matrix form multiplied by $W_V$ we get
\begin{equation}
    Z = \text{softmax}(\frac{1}{\sqrt{d_k}}QK^T)V = \text{softmax}(\beta RW_QW_K^TY^T)YW_KW_V.
\end{equation}

We can recognize the left part of the equation as being the attention mechanism, while the right side is the update rule for MHNs, followed by a matrix multiplication with $W_V$.

\section{Theoretical Proofs}
\subsection{Proof of Theorem \ref{thm:separation}}\label{app:thm:separation}

\begin{proof}
Let,
\begin{equation}
    \mathcal{S}(\mu, q) := \langle \zeta^\mu, q \rangle + \lambda\, s_\mu
\end{equation}
At the retrieval fixed point, $r_k^* = \mathcal{S}(k, q)$ for each memory $k$. The retrieval probability for the target $\mu$ is the softmax output:
\begin{equation}
    P(\mu) = \hat{r}_\mu = \frac{e^{\beta\, \mathcal{S}(\mu, q)}}{\displaystyle\sum_{k=1}^{M} e^{\beta\, \mathcal{S}(k, q)}}.
\end{equation}
Dividing numerator and denominator by $e^{\beta\, \mathcal{S}(\mu, q)}$:
\begin{equation}
    P(\mu) = \frac{1}{1 + \displaystyle\sum_{\nu \neq \mu} e^{-\beta\, \Delta_\nu}},
\end{equation}
where $\Delta_\nu := \mathcal{S}(\mu, q) - \mathcal{S}(\nu, q) \geq \Delta$ for all $\nu \neq \mu$, by definition of $\Delta$ as the minimum gap. Since the exponential is monotonically decreasing, $e^{-\beta\, \Delta_\nu} \leq e^{-\beta\, \Delta}$ for each $\nu \neq \mu$. Summing over the $M - 1$ distractor terms:
\begin{equation}\label{eq:exact_bound}
    P(\mu) \geq \frac{1}{1 + (M-1)\,e^{-\beta\, \Delta}},
\end{equation}
To derive the threshold condition, we require $P(\mu) \geq 1 - \epsilon$, i.e.:
\begin{equation}
    \frac{1}{1 + (M-1)\,e^{-\beta\, \Delta}} \geq 1 - \epsilon.
\end{equation}
Inverting and rearranging:
\begin{equation}
    1 + (M-1)\,e^{-\beta\, \Delta} \leq \frac{1}{1-\epsilon},
\end{equation}
\begin{equation}
    (M-1)\,e^{-\beta\, \Delta} \leq \frac{1}{1-\epsilon} - 1 = \frac{\epsilon}{1-\epsilon}.
\end{equation}
Taking logarithms:
\begin{equation}
    -\beta\, \Delta \leq \ln\!\left(\frac{\epsilon}{(1-\epsilon)(M-1)}\right),
\end{equation}
which leads to the final result.
\end{proof}

\subsection{Formal version of Theorem \ref{thm:wta} and Proof}\label{app:thm:wta}
\begin{theorem}
Let $\alpha_{\mathrm{crit}}\triangleq (1-\mu_{\min})^{-1}$, and for state $s(t)$ and input $u$, define projections on the eigenvectors as $a_k(t)=v_k^\top s(t)$ and $c_k=v_k^\top u$ respectively. Finally, define $p(t)=\mathrm{softmax}(\beta s(t))$ as the distribution induced by $s(t)$ post softmax.

\textbf{(i) Subcritical ($\alpha<\alpha_{\mathrm{crit}}$).}
Then, $A(\alpha)\succ0$ and the fixed point solution is unique:
\[
s^*=A(\alpha)^{-1}u=\sum_{k=1}^N \frac{c_k}{\eta_k(\alpha)}\,v_k,
\qquad
p^*=\mathrm{softmax}(\beta s^*).
\]
The modes $v_k\in \mathcal{V}_{\min}$ have the maximum possible amplification which is given by $\frac{1}{\eta_{\min}(\alpha)}$. Furthermore, $\eta_{\min}(\alpha)\to0$ as $\alpha\to \alpha_{\mathrm{crit}}$ leading to larger amplifications.

\textbf{(ii) Supercritical ($\alpha>\alpha_{\mathrm{crit}}$).}
Define
$\gamma_{k}(\alpha)\triangleq\alpha(1-\mu_{k})-1$,
and $\xi \triangleq P_{\min}\sum_{\{k : \mu_k=\mu_{\min}\}} w_k\,v_k$, with $w_k\triangleq a_k(0)+c_k/\gamma_k(\alpha)$.
\[
\xi\triangleq P_{\min}\,w,
\qquad
w\triangleq\sum_{\mu_k=\mu_{\min}} w_k\,v_k,
\qquad
w_k\triangleq a_k(0)+\frac{c_k}{\gamma_{\min}(\alpha)}.
\]
Assume $\xi\neq 0$ and $i^*:=\arg\max_{i\in[N]} \xi_i$ is unique. Then
\[
p(t)\to\delta_{i,i^*}\quad\text{as }t\to\infty.
\]
In particular, for any $\delta\in(0,1)$, $p_{i^*}(t^*(\alpha))\ge 1-\delta$ with
\[
t^*(\alpha)=\frac{\tau_s}{\gamma_{\min}(\alpha)}
\Bigl[\ln\!\Bigl(\frac{\kappa}{\beta\,\Delta_{\xi}}\Bigr)+O(1)\Bigr],
\quad
\Delta_{\xi}:=\xi_{i^*}-\max_{j\neq i^*} \xi_j>0,
\]
where $\kappa>0$ is independent of $\alpha$.
\end{theorem}
\begin{proof}
Let $N$ unit-norm memories $\{\zeta^\mu\}_{\mu=1}^N \subset \mathbb{R}^d$ with Gram matrix $G\in\mathbb{R}^{N\times N}$, $G_{ij}=\langle \zeta^i,\zeta^j\rangle$, so $G_{ii}=1$. Define $W_{h_2}:=G-I$ and
\begin{equation}\label{eq:app_A_def}
A(\alpha):=I+\alpha W_{h_2}=(1-\alpha)I+\alpha G.
\end{equation}
Let $R:=\mathrm{rank}(G)=\min(N,d)$ and write the spectral decomposition
\begin{equation}\label{eq:app_G_spec}
G=\sum_{k=1}^{R}\mu_k\,v_k v_k^\top,
\qquad
\mu_1\ge\cdots\ge\mu_R>0,
\end{equation}
with orthonormal $\{v_k\}_{k=1}^R$. When $N>R$, extend $\{v_k\}_{k=1}^R$ to an orthonormal basis $\{v_k\}_{k=1}^N$ of $\mathbb{R}^N$ by choosing
$\{v_k\}_{k=R+1}^N$ to span $\ker(G)$, and set $\mu_k:=0$ for $k>R$. Define the smallest eigenvalue
\begin{equation}\label{eq:app_mu_min_def}
\mu_{\min}:=\mu_N
=
\begin{cases}
\mu_N>0 & \text{if } N\le d,\\
0 & \text{if } N>d,
\end{cases}
\end{equation}
and let $\mathcal{V}_{\min}:=\mathrm{span}\{v_k:\mu_k=\mu_{\min}\}$ denote its eigenspace with orthogonal projector
\begin{equation}\label{eq:app_Pmin_def}
P_{\min}:=\sum_{\mu_k=\mu_{\min}} v_k v_k^\top.
\end{equation}
Then
\begin{equation}\label{eq:app_eta_def}
A(\alpha)v_k=\eta_k(\alpha)v_k,
\qquad
\eta_k(\alpha)=1-\alpha(1-\mu_k),
\qquad k=1,\ldots,N,
\end{equation}
so the smallest eigenvalue is
\begin{equation}\label{eq:app_eta_min_def}
\eta_{\min}(\alpha):=\min_k \eta_k(\alpha)=1-\alpha(1-\mu_{\min}),
\end{equation}
attained for all $k$ with $\mu_k=\mu_{\min}$. Define
\begin{equation}\label{eq:app_alpha_crit}
\alpha_{\mathrm{crit}}:=\frac{1}{1-\mu_{\min}}.
\end{equation}
Write the input as $u=\sum_{k=1}^N c_k v_k$ with $c_k:=v_k^\top u$ and recall $p(t)=\mathrm{softmax}(\beta s(t))$.

\textbf{Subcritical case ($\alpha<\alpha_{\mathrm{crit}}$).}
Since $\alpha<\alpha_{\mathrm{crit}}$ iff $\eta_{\min}(\alpha)>0$, we have $\eta_k(\alpha)\ge \eta_{\min}(\alpha)>0$ for all $k$, hence $A(\alpha)\succ 0$ and is invertible. The fixed point is unique and satisfies
\begin{equation}\label{eq:app_fp}
s^*=A(\alpha)^{-1}u
=\sum_{k=1}^N \frac{c_k}{\eta_k(\alpha)}\,v_k,
\qquad
p^*=\mathrm{softmax}(\beta s^*).
\end{equation}
Thus mode $k$ is amplified by $1/\eta_k(\alpha)$, and the modes in $\mathcal{V}_{\min}$ are maximally amplified by
$1/\eta_{\min}(\alpha)$ with $\eta_{\min}(\alpha)\to 0$ as $\alpha\to \alpha_{\mathrm{crit}}$.

\textbf{Supercritical case ($\alpha>\alpha_{\mathrm{crit}}$).}
Consider the gate dynamics
\begin{equation}\label{eq:app_gate_dyn}
\tau_s\,\dot{s}(t)=u-A(\alpha)s(t),
\qquad
p(t)=\mathrm{softmax}(\beta s(t)),
\end{equation}
and define modal coordinates $a_k(t):=v_k^\top s(t)$. Projecting equation~\eqref{eq:app_gate_dyn} onto $v_k$ gives
\begin{equation}\label{eq:app_scalar_ode}
\tau_s\,\dot a_k(t)=c_k-\eta_k(\alpha)a_k(t),
\end{equation}
so
\begin{equation}\label{eq:app_mode_sol}
a_k(t)=\frac{c_k}{\eta_k(\alpha)}+\Bigl(a_k(0)-\frac{c_k}{\eta_k(\alpha)}\Bigr)e^{-\eta_k(\alpha)t/\tau_s}.
\end{equation}
Since $\alpha>\alpha_{\mathrm{crit}}$, $\eta_{\min}(\alpha)<0$ on $\mathcal{V}_{\min}$. Set
\begin{equation}\label{eq:app_gamma_min_def}
\gamma_{\min}(\alpha):=-\eta_{\min}(\alpha)=\alpha(1-\mu_{\min})-1>0.
\end{equation}
For each $k$ with $\mu_k=\mu_{\min}$ (hence $\eta_k(\alpha)=\eta_{\min}(\alpha)$), rewrite equation~\eqref{eq:app_mode_sol} as
\begin{equation}\label{eq:app_unstable_form}
a_k(t)=-\frac{c_k}{\gamma_{\min}(\alpha)}+w_k\,e^{\gamma_{\min}(\alpha)t/\tau_s},
\qquad
w_k:=a_k(0)+\frac{c_k}{\gamma_{\min}(\alpha)}.
\end{equation}
All modes with $\mu_k>\mu_{\min}$ satisfy $\eta_k(\alpha)>\eta_{\min}(\alpha)$ and, in particular, remain stable in a neighborhood of
$\alpha_{\mathrm{crit}}$, contributing bounded terms to $s(t)$.

Expanding $s(t)=\sum_{k=1}^N a_k(t)v_k$ and collecting all stable-mode contributions and constant offsets into $r(t)$ yields, for each coordinate $i$,
\begin{equation}\label{eq:app_dominance}
s_i(t)=\xi_i\,e^{\gamma_{\min}(\alpha)t/\tau_s}+r_i(t),
\quad
\xi_i:=\sum_{\mu_k=\mu_{\min}} w_k\,(v_k)_i=\bigl(P_{\min}w\bigr)_i,
\end{equation}
where $w:=\sum_{\mu_k=\mu_{\min}} w_k v_k\in\mathcal{V}_{\min}$ and $|r_i(t)|\le C$ for all $t$, for a constant $C$ independent of $t$.
Assume $\xi:=(\xi_i)_{i=1}^N\neq 0$ and choose
\begin{equation}\label{eq:app_i_star}
i^*:=\arg\max_{i\in[N]} \xi_i,
\end{equation}
with the maximizer unique. For any $j\neq i^*$, subtracting equation~\eqref{eq:app_dominance} gives
\begin{equation}\label{eq:app_gap_growth}
s_{i^*}(t)-s_j(t)=(\xi_{i^*}-\xi_j)e^{\gamma_{\min}(\alpha)t/\tau_s}+O(1)\xrightarrow[t\to\infty]{}+\infty,
\end{equation}
since $\xi_{i^*}>\xi_j$. Then
\begin{equation}\label{eq:app_softmax_vanish}
p_j(t)=\frac{e^{\beta s_j(t)}}{\sum_{\ell}e^{\beta s_\ell(t)}}
\le \frac{e^{\beta s_j(t)}}{e^{\beta s_{i^*}(t)}}
=e^{-\beta(s_{i^*}(t)-s_j(t))}\xrightarrow[t\to\infty]{}0,
\end{equation}
so $p(t)\to \delta_{i,i^*}$, proving the WTA claim.

Let $j^*\in\arg\max_{j\neq i^*} \xi_j$ be a strongest competitor and define $\Delta_{\xi}:=\xi_{i^*}-\xi_{j^*}>0$.
The softmax identity gives
\begin{equation}\label{eq:app_pi_star_bound}
p_{i^*}(t)
=\frac{1}{1+\sum_{j\neq i^*}e^{-\beta(s_{i^*}(t)-s_j(t))}}
\ge
\frac{1}{1+(N-1)e^{-\beta(s_{i^*}(t)-s_{j^*}(t))}}.
\end{equation}
Thus $p_{i^*}(t)\ge 1-\delta$ is ensured by
\begin{equation}\label{eq:app_threshold_L}
\beta\bigl(s_{i^*}(t)-s_{j^*}(t)\bigr)\ge
\ln\!\Bigl(\frac{(1-\delta)(N-1)}{\delta}\Bigr)=:L_\delta.
\end{equation}
By equation~\eqref{eq:app_dominance}, the leading gap satisfies
\begin{equation}\label{eq:app_gap_leading}
s_{i^*}(t)-s_{j^*}(t)
=
\Delta_{\xi}\,e^{\gamma_{\min}(\alpha)t/\tau_s}+O(1).
\end{equation}
Solving equation~\eqref{eq:app_threshold_L} using equation~\eqref{eq:app_gap_leading} and absorbing additive constants into an $O(1)$ bracket yields, for some
$\kappa>0$ independent of $\alpha$,
\begin{equation}\label{eq:app_tstar_final}
t^*(\alpha)=\frac{\tau_s}{\gamma_{\min}(\alpha)}
\Bigl[\ln\!\Bigl(\frac{\kappa}{\beta\,\Delta_{\xi}}\Bigr)+O(1)\Bigr],
\end{equation}
which is exactly the claimed form.

Finally, when $N\le d$ and $\mu_N$ is simple, $\mathcal{V}_{\min}=\mathrm{span}\{v_N\}$ so $\xi_i=w_N(v_N)_i$ and
$i^*=\arg\max_i w_N(v_N)_i$. When $N>d$, $\mu_{\min}=0$ and
$\gamma_{\min}(\alpha)=\alpha-1$, while $P_{\min}$ is the projector onto $\ker(G)$, so the winner is determined by the projection of
$w$ onto $\ker(G)$.
\end{proof}

\subsection{Criticality with $\lambda > 0$}\label{app:thm:lambda}
\begin{theorem}[Critical Regime with $\lambda > 0$]\label{thm:lambda_crit}
Let $\lambda > 0$ and $\tau_s \ll \tau_r$ so that the gates settle to $s^* = A(\alpha)^{-1}(u + \lambda\hat{r})$ on the fast timescale. Further assume that the context input is uniform (i.e. $u = u_0 \mathbf{1}$), the query provides no bias (i.e. $\langle \zeta^\mu, \hat{q}\rangle = q_0$ for all $\mu \in [N]$), and $G\mathbf{1} = \mu_G \mathbf{1}$ for some $\mu_G > 0$. 

Let $\mu_{\min}^{\perp}$ denote the smallest eigenvalue of $G$ restricted to the subspace $\mathbf{1}^{\perp} \triangleq \{v \in \mathbb{R}^N : \langle v, \mathbf{1}\rangle = 0\}$, and define $\alpha_{\mathrm{crit}} \triangleq (1 - \mu_{\min}^{\perp})^{-1}$.
Then, provided $\lambda^2\beta < N$, the \emph{uniform retrieval state} $\hat{r} = \frac{1}{N}\mathbf{1}$ is a fixed point of the slow retrieval dynamics whose stability is governed by the threshold
\[
\;\alpha_{\mathrm{crit}}^{\lambda} \;=\; \alpha_{\mathrm{crit}}\!\left(1 - \frac{\lambda^2\beta}{N}\right).
\]
The uniform fixed point is stable for $\alpha < \alpha_{\mathrm{crit}}^{\lambda}$ and unstable for $\alpha > \alpha_{\mathrm{crit}}^{\lambda}$.

In particular, for $\alpha_{\mathrm{crit}}^{\lambda} < \alpha < \alpha_{\mathrm{crit}}$, the gate subsystem alone is \emph{subcritical} (its uniform state is stable), and it is the retrieval feedback loop through the softmax nonlinearity that destabilizes the uniform state, inducing winner-take-all retrieval.
\end{theorem}

\begin{proof}
Recall $A(\alpha) = (1-\alpha)I + \alpha G$. By assumption, $\mathbf{1}$ is an eigenvector of $G$ with eigenvalue $\mu_{G}$, hence of $A(\alpha)$ with eigenvalue
\begin{equation}\label{eq:eta_G}
\eta_{\scriptscriptstyle G}(\alpha) = 1 - \alpha(1 - \mu_{\scriptscriptstyle G}).
\end{equation}
Observe that $G$ is symmetric, so all eigenvectors other than $\mathbf{1}$ lie in $\mathbf{1}^\perp$. Thus, for $\alpha < \alpha_{\mathrm{crit}}$ all eigenvalues $\eta_k(\alpha) = 1-\alpha(1-\mu_k)$ are positive (since $\mu_k \geq \mu_{\min}^{\perp}$), so $A(\alpha)$ is invertible with $A(\alpha)^{-1}\mathbf{1} = \eta_{\scriptscriptstyle G}^{-1}\,\mathbf{1}$.

At the uniform retrieval state $\hat{r} = \frac{1}{N}\mathbf{1}$, the gate steady-state is
\[
s^* = A(\alpha)^{-1}\!\bigl(u_0\mathbf{1} + \tfrac{\lambda}{N}\mathbf{1}\bigr) = \frac{u_0 + \lambda/N}{\eta_{\scriptscriptstyle G}(\alpha)}\,\mathbf{1}.
\]
Substituting into the retrieval dynamics equation~\eqref{eq:r_dyn} and absorbing the constant query bias $q_0\mathbf{1}$ into a scalar, the effective slow dynamics can be written as the vector field
\begin{equation}\label{eq:retrieval_vf}
\tau_r \dot r \;=\; \underbrace{\lambda\,A(\alpha)^{-1}u}_{\text{constant bias}}
\;+\;\lambda^2\,A(\alpha)^{-1}\,\hat r
\;+\;q_0\,\mathbf{1}
\;-\;r,
\end{equation}
where $\hat r=\mathrm{softmax}(\beta r)$. Observe that $r^* = c\mathbf{1}$ where 
\[
c = \frac{\lambda u_0}{\eta_G(\alpha)} + \frac{\lambda^2}{N\eta_G(\alpha)}+q_0,
\]
is indeed a fixed point of the dynamics.

The stability of this system is determined by the Jacobian of the vector field $F(r) = \lambda^2 A(\alpha)\hat{r} - r$ evaluated at the uniform fixed point. 
The Jacobian of the softmax at $\hat{r} = \frac{1}{N}\mathbf{1}$ is
\[
\mathcal{J} \;\triangleq\; \frac{\partial\,\hat{r}}{\partial\, r}\bigg|_{r=c\mathbf{1}} = \beta\!\left(\mathrm{diag}(\hat{r}) - \hat{r}\hat{r}^{\!\top}\right) = \frac{\beta}{N}\!\left(I - \frac{1}{N}\mathbf{1}\mathbf{1}^{\!\top}\right) \;=\; \frac{\beta}{N}\,P_{\perp},
\]
where $P_{\perp} = I - \frac{1}{N}\mathbf{1}\mathbf{1}^{\!\top}$ is the orthogonal projector onto $\mathbf{1}^{\perp}$. The Jacobian of the vector field equation~\eqref{eq:retrieval_vf} at $r=c\mathbf{1}$ is therefore
\begin{equation}\label{eq:J_sys}
J_{\mathrm{sys}} \;=\; \frac{\lambda^2\beta}{N}\,A(\alpha)^{-1}\,P_{\perp} \;-\; I.
\end{equation}

Let $\{v_k\}_{k=1}^{N}$ be an orthonormal eigenbasis of $G$ with $Gv_k = \mu_k v_k$. By assumption, one eigenvector is $v_1 = \mathbf{1}/\sqrt{N}$ with eigenvalue $\mu_{\scriptscriptstyle G}$, and the remaining $\{v_k\}_{k=2}^{N}$ span $\mathbf{1}^{\perp}$.

For any $v_k \in \mathbf{1}^{\perp}$:
\begin{align}
    P_{\perp}\, v_k &= v_k, \label{eq:proj_vk}\\
    A(\alpha)\, v_k &= \eta_k(\alpha)\, v_k, \quad \text{where } \eta_k(\alpha) = 1 - \alpha(1-\mu_k). \label{eq:A_vk}
\end{align}
By multiplying both sides of equation~\eqref{eq:J_sys} by $v_k$ we get:
\begin{equation}\label{eq:JF_eig}
J_{\mathrm{sys}}\, v_k \;=\; \left(\frac{\lambda^2\beta}{N\,\eta_k(\alpha)} - 1\right)v_k, \qquad k = 2,\ldots,N.
\end{equation}

The uniform fixed point is unstable when $J_{\mathrm{sys}}$ has a positive eigenvalue. From equation~\eqref{eq:JF_eig}, the eigenvalue is maximized for the mode with the \emph{smallest} $\eta_k(\alpha)$, which corresponds to $\mu_{\min}^{\perp}$:
\begin{equation}\label{eq:instability_cond}
\frac{\lambda^2\beta}{N\,\eta_{\min}^{\perp}(\alpha)} - 1 > 0
\quad\Longleftrightarrow\quad
\eta_{\min}^{\perp}(\alpha) < \frac{\lambda^2\beta}{N}.
\end{equation}
Substituting $\eta_{\min}^{\perp}(\alpha) = 1 - \alpha(1 - \mu_{\min}^{\perp})$ and solving for $\alpha$:
\[
1 - \alpha(1-\mu_{\min}^{\perp}) < \frac{\lambda^2\beta}{N} \quad\Longrightarrow\quad \alpha > \frac{1 - \lambda^2\beta/N}{1-\mu_{\min}^{\perp}} = \alpha_{\mathrm{crit}}\!\left(1 - \frac{\lambda^2\beta}{N}\right).
\]
The requirement $\lambda^2\beta < N$ ensures $\alpha_{\mathrm{crit}}^{\lambda} > 0$: when $\lambda^2\beta \geq N$, the feedback gain is strong enough that the uniform state is unstable for \emph{all} $\alpha > 0$.
\end{proof}
\subsection{Formal Statement of Theorem~\ref{thm:self_consistent} and Proof}\label{app:self_consistent}
\begin{theorem}\label{thm:appendix_self_consistent}
Assume $A(\alpha)>0$, $\eta_{\min}(\alpha):=\lambda_{\min}(A(\alpha)).$ For any $\lambda>0$, every fixed point $(s^*,r^*)$ of the coupled gate-retrieval subsystem satisfies,
\begin{equation}
    s^*=A(\alpha)^{-1}(u+\lambda p^*),\quad r^*=b+ \lambda A(\alpha)^{-1} u + \lambda ^2 A(\alpha)^{-1} p^*,
    \label{eq:app_comb_fix_point}
\end{equation}
where $p^*=\mathrm{softmax}_\beta(r^*)$. Equivalently, $p^*$ is a fixed point of the self-consistent map,
\begin{equation}\label{eq:self-consistent-map-app}
    p^*=\Phi_{\alpha,\lambda}(p^*):= \mathrm{softmax}_\beta \left(b+ \lambda A(\alpha)^{-1} u+\lambda^2 A(\alpha)^{-1} p^* \right)
\end{equation}
Conversely, any fixed point $p^*=\Phi_{\alpha,\lambda}(p^*)$ induces a fixed point $(s^*,r^*)$ through \eqref{eq:app_comb_fix_point}. Furthermore, if 
\begin{equation}\label{eq:contraction-condition}
    \frac{\beta\lambda^2}{2\eta_{\min}(\alpha)}<1,
\end{equation} then $\Phi_{\alpha,\lambda}$ is a contraction map on the probability simplex over $\mathbb{R}^N$,
$$
\Delta_N := \{p \in \mathbb{R}^n \mid p_i\geq 0, \sum_i p_i = 1\},
$$
under the Euclidean norm. Consequently, the subsystem has a unique fixed point.
\end{theorem}

We will use the following lemma in the proof of Theorem~\ref{thm:appendix_self_consistent}.
\begin{lemma}
\label{lem:softmax-jacobian-bound}
For every $p\in\Delta_N$,
\begin{equation}
    \left\|\mathrm{diag}(p)-pp^\top\right\|_2\le \frac12.
    \label{eq:covariance-bound}
\end{equation}
Consequently, for $\mathrm{softmax}_\beta$,
\begin{equation}
    \left\|D(p)\right\|_2
    =\left\|\beta\bigl(\mathrm{diag}(p)-pp^\top\bigr)\right\|_2
    \le \frac{\beta}{2}.
    \label{eq:softmax-jacobian-beta-over-two}
\end{equation}
\end{lemma}
\begin{proof}
    Let
\[
    C(p):=\mathrm{diag}(p)-pp^\top.
\]
First, $C(p)$ is symmetric positive semidefinite because, for any
$x\in\mathbb{R}^N$,
\begin{align}
    x^\top C(p)x
    &=\sum_{i=1}^N p_i x_i^2
      -\left(\sum_{i=1}^N p_i x_i\right)^2
      \notag\\
    &=\mathrm{Var}_{i\sim p}(x_i)
    \ge 0.
    \label{eq:rayleigh-variance-identity}
\end{align}
Therefore,
\[
    \|C(p)\|_2
    =\sup_{\|x\|_2=1} x^\top C(p)x.
\]
Fix any $x\in\mathbb{R}^N$ with $\|x\|_2=1$. By the variance identity
\eqref{eq:rayleigh-variance-identity},
\[
    x^\top C(p)x=\mathrm{Var}_{i\sim p}(x_i).
\]
Popoviciu's variance inequality gives
\[
    \mathrm{Var}_{i\sim p}(x_i)
    \le
    \frac14\left(\max_i x_i-\min_i x_i\right)^2.
\]
Let $a:=\max_i x_i$ and $b:=\min_i x_i$. Then
\[
    (a-b)^2\le 2a^2+2b^2=2(a^2+b^2).
\]
Moreover, since $a$ and $b$ are two coordinates of $x$,
\[
    a^2+b^2\le \|x\|_2^2=1,
\]
where the inequality is also valid when the maximum and minimum occur at the
same coordinate. Hence
\[
    \left(\max_i x_i-\min_i x_i\right)^2\le 2.
\]
Thus
\[
    x^\top C(p)x
    =\mathrm{Var}_{i\sim p}(x_i)
    \le \frac12.
\]
Since $x$ was arbitrary, this proves equation~\eqref{eq:covariance-bound}. Multiplying by $\beta$ proves
equation~\eqref{eq:softmax-jacobian-beta-over-two}.
\end{proof}
We now begin the proof of Theorem~\ref{thm:appendix_self_consistent}
\begin{proof}[Proof of Theorem~\ref{thm:appendix_self_consistent}]
    At a fixed point of equation~\eqref{eq:s_dyn},
\[
    0=u-A(\alpha)s^*+\lambda p^*.
\]
Since $A(\alpha)\succ0$, it is invertible, so
\[
    s^*=A(\alpha)^{-1}(u+\lambda p^*).
\]
At a fixed point of equation~\eqref{eq:r_dyn},
\[
    0=b+\lambda s^*-r^*,
\]
and therefore
\[
    r^*=b+\lambda s^*.
\]
Substituting the expression for $s^*$ gives
\[
    r^*=b+\lambda A(\alpha)^{-1}u+
    \lambda^2A(\alpha)^{-1}p^*.
\]
Finally, by definition, $p^*=\mathrm{softmax}_\beta(r^*)$, hence
\[
    p^*
    =\mathrm{softmax}_\beta\bigl(b+\lambda A(\alpha)^{-1}u+
    \lambda^2A(\alpha)^{-1}p^*\bigr),
\]
which is exactly equation~\eqref{eq:self-consistent-map-app}. Conversely, if
$p^*=\Phi_{\alpha,\lambda}(p^*)$, define $s^*$ and $r^*$ by equation~\eqref{eq:app_comb_fix_point}. Direct substitution verifies both fixed-point
equations.

Now we show that $\Phi_{\alpha, \lambda}$ is a contraction mapping. By Lemma~\ref{lem:softmax-jacobian-bound}, the Jacobian of $\mathrm{softmax}_\beta$ obeys
\[
    \|D(\rho)\|_2\le \frac{\beta}{2}
    \qquad \text{for all }\rho\in\Delta_N.
\]
Now fix arbitrary $p,p'\in\Delta_N$. Because $\Delta_N$ is convex, the line
segment
\[
    p_t:=p'+t(p-p'),\qquad t\in[0,1],
\]
lies in $\Delta_N$. Define
\[
    z_t:=h_{\alpha,\lambda}+B_{\alpha,\lambda}p_t,
    \qquad
    \rho_t:=\mathrm{softmax}_\beta(z_t).
\]
By the Fundamental Theorem of Calculus and the chain rule,
\begin{align}
    \Phi_{\alpha,\lambda}(p)-\Phi_{\alpha,\lambda}(p')
    &=\int_0^1
    D(\rho_t)B_{\alpha,\lambda}(p-p')\,dt.
    \label{eq:ftc-contraction}
\end{align}
Taking Euclidean norms and using submultiplicativity,
\begin{align}
    \|\Phi_{\alpha,\lambda}(p)-\Phi_{\alpha,\lambda}(p')\|_2
    &\le
    \int_0^1
    \|D(\rho_t)\|_2\,
    \|B_{\alpha,\lambda}\|_2\,
    \|p-p'\|_2\,dt
    \notag\\
    &\le
    \frac{\beta}{2}\,
    \|B_{\alpha,\lambda}\|_2\,
    \|p-p'\|_2.
    \label{eq:lipschitz-before-B}
\end{align}
Since
\[
    B_{\alpha,\lambda}=\lambda^2A(\alpha)^{-1},
\]
and $A(\alpha)$ is symmetric positive definite, its inverse is symmetric
positive definite and
\[
    \|A(\alpha)^{-1}\|_2
    =\lambda_{\max}(A(\alpha)^{-1})
    =\frac{1}{\lambda_{\min}(A(\alpha))}
    =\frac{1}{\eta_{\min}(\alpha)}.
\]
Therefore,
\[
    \|B_{\alpha,\lambda}\|_2
    =\lambda^2\|A(\alpha)^{-1}\|_2
    =\frac{\lambda^2}{\eta_{\min}(\alpha)}.
\]
Substituting this into \eqref{eq:lipschitz-before-B} gives
\[
    \|\Phi_{\alpha,\lambda}(p)-\Phi_{\alpha,\lambda}(p')\|_2
    \le
    \frac{\beta\lambda^2}{2\eta_{\min}(\alpha)}
    \|p-p'\|_2.
\]
Under condition \eqref{eq:contraction-condition}, the Lipschitz constant is
strictly smaller than one. Hence $\Phi_{\alpha,\lambda}$ is a contraction map on
$\Delta_N$. Since $\Delta_N$ is closed in the complete normed space
$(\mathbb{R}^N,\|\cdot\|_2)$, it is complete. Also,
$\Phi_{\alpha,\lambda}(\Delta_N)\subseteq\Delta_N$ because softmax outputs a
probability vector. The Banach fixed-point theorem therefore implies that
$\Phi_{\alpha,\lambda}$ has a unique fixed point $p^*\in\Delta_N$. By
\ref{thm:appendix_self_consistent}, this unique $p^*$ induces a unique fixed point
$(s^*,r^*)$ of the coupled gate--retrieval subsystem.
\end{proof}

Notice that since
\[
    \eta_{\min}(\alpha)=1-\alpha(1-\mu_{\min}),
\]
the contraction condition becomes harder to satisfy as
$\alpha\to\alpha_{\rm crit}$. Thus, even for fixed $\lambda>0$, the retrieval
feedback term $\lambda^2A(\alpha)^{-1}p$ becomes increasingly amplified as
the gate subsystem approaches its own critical point.

\section{Experimental Details}\label{app:experiments}
All experiments were conducted on 1 NVIDIA H100.
\subsection{Context-Augmented Memory Separation (Figure~\ref{fig:separation})}\label{app:exp:separation}

\textbf{Stored memories.}
We fix $N=50$ unit-norm memories in $\mathbb{R}^{d}$ with $d=10$ by sampling i.i.d.\ Gaussian vectors and normalizing:
\[
\zeta^\mu \sim \mathcal{N}(0,I_d),\qquad \zeta^\mu \leftarrow \zeta^\mu/\|\zeta^\mu\|_2,\qquad \mu\in[N].
\]
The inverse temperature is $\beta=3.5$.

\textbf{Query and context noise model.}
Each trial selects a target index $\mu$ uniformly from $[N]$ and draws
\[
q=\zeta^\mu+\sigma_q\,\varepsilon,\qquad \varepsilon\sim\mathcal{N}(0,I_d),
\]
\[
c=\zeta^\mu+\sigma_c\,\eta,\qquad \eta\sim\mathcal{N}(0,I_d),\qquad \sigma_c=0.3.
\]
We sweep $\sigma_q$ over 12 evenly spaced noise levels in $[0,3]$.

\textbf{Metrics and plotting quantities.}
Retrieval accuracy is the fraction of trials for which $\arg\max_{\nu\in[N]}\hat r_\nu=\mu$.
For each trial we also compute the effective separation gap
\[
\Delta = S_\mu(q;\lambda) - \max_{\nu\neq\mu} S_\nu(q;\lambda)
\]
and the target probability $\hat r_\mu$.
We plot:
(i) accuracy vs.\ $\sigma_q$ for each $\lambda\in\{0,1,2,4,8,16\}$,
(ii) a scatter of $(\Delta,\hat r_\mu)$ over all trials, overlaid with the exact logistic lower bound
\[
\hat r_\mu \ge \frac{1}{1+(N-1)e^{-\beta\Delta}},
\]
and (iii) box plots of $\Delta$ at fixed noise $\sigma_q\approx 1.0$ (we collect trials with $|\sigma_q-1.0|<0.15$).
Each $(\sigma_q,\lambda)$ pair uses 6000 trials (so $12\times 6000\times 5$ total softmax evaluations). Finally, we set $\alpha=0.1$.

\subsection{Phase Transition in Gate Selectivity (Figure~\ref{fig:competition})}\label{app:exp:wta}

\textbf{Single-cluster memory model.}
Each trial generates a single cluster of $N=50$ unit-norm memories in $\mathbb{R}^{d}$ with $d=10$ as follows.
Sample a random centroid $g\sim\mathcal{N}(0,I_d)$, normalize, and scale to norm $2$:
\[
g\leftarrow 2\,g/\|g\|_2.
\]
Then for each $\mu\in[N]$ draw
\[
\zeta^\mu = g + \sigma\,\xi^\mu,\qquad \xi^\mu\sim\mathcal{N}(0,I_d),\qquad \sigma=0.3,
\]
and normalize $\zeta^\mu\leftarrow \zeta^\mu/\|\zeta^\mu\|_2$.

\textbf{Competition matrix, input, and output.}
Let $G=\zeta^\top\zeta$ and define $W_{h_2}=G-I$ (equivalently, set the diagonal of $G$ to zero).
For each penalization strength $\alpha$ we form
\[
A(\alpha)=I+\alpha W_{h_2},
\qquad
u=\frac{1}{N}\mathbf{1},
\qquad
p^*(\alpha)=\mathrm{softmax}\!\bigl(\beta s(\alpha)\bigr),
\]
with $\beta=3.5$.
When $A(\alpha)$ is numerically stable (see below), we use the exact fixed point $s(\alpha)=A(\alpha)^{-1}u$.

\textbf{Estimating the critical point.}
To estimate $\mu_N=\lambda_{\min}(G)$, we sample 1000 independent clusters, compute the smallest eigenvalue of each Gram matrix $G$, and report the mean $\mu_N$.
We then set
\[
\alpha_{\mathrm{crit}}=\frac{1}{1-\mu_N}.
\]

\textbf{Sweep over $\alpha$ and number of trials.}
We run $1000$ independent trials per $\alpha$ (seeds $0,\ldots,999$).
The sweep uses:
(i) a coarse grid of 17 points on $[0,2]$ and
(ii) a dense grid of 60 points between $\min\{\alpha_{\mathrm{crit}},1.25\}$ and $\max\{\alpha_{\mathrm{crit}},1.25\}$,
plus the value $\alpha_{\mathrm{crit}}$ itself; duplicates are removed after rounding to 4 decimals.

\textbf{Stable vs.\ unstable handling (implementation detail).}
Let $\lambda_{\min}(A(\alpha))$ be the smallest eigenvalue of $A(\alpha)$.
If $\lambda_{\min}(A(\alpha))>0.01$, we compute $s(\alpha)=A(\alpha)^{-1}u$.
Otherwise, we evaluate the closed-form transient solution of the linear ODE $\dot s=u-A(\alpha)s$ at a finite horizon $t_{\mathrm{end}}$ (chosen adaptively as per Theorem \ref{thm:wta}), and set $s(\alpha)=s(t_{\mathrm{end}})$; this produces a numerically well-defined proxy that approaches the WTA behavior in the unstable regime.

\textbf{Reported quantity.}
For each $\alpha$, we report the mean peak probability
\[
\mathbb{E}\Bigl[\max_{i\in[N]} p^*_i(\alpha)\Bigr]
\]
over the 1000 trials, and visualize the sharp transition around $\alpha_{\mathrm{crit}}$ (with a zoomed-in panel near the transition).

\subsection{Coupled Gate-Retrieval Dynamics (Figure~\ref{fig:competition}, Panel d)}\label{app:exp:coupled}

\textbf{Model and coupling mechanism.}
We extend the gate-only model (Section~\ref{app:exp:wta}) by coupling retrieval activations back through the gates via a cross-circuit parameter $\lambda\geq0$.
The gate dynamics settle fast ($\tau_s\ll\tau_r$), yielding the quasi-static fixed point
\[
s^*(\alpha,\lambda) = A(\alpha)^{-1}\bigl(u + \lambda\,\hat{r}\bigr),
\]
where $\hat{r}=\mathrm{softmax}(\beta r)$ are the retrieval probabilities.
The retrieval activations $r$ then satisfy the self-consistent equation
\[
r^* = \lambda\,A(\alpha)^{-1}u + \lambda^2\,A(\alpha)^{-1}\hat{r}^*.
\]

\textbf{Fixed-point iteration.}
For each $(\alpha,\lambda)$ pair, we solve for $r^*$ by iterating
\[
r^{(k+1)} = \lambda\,A(\alpha)^{-1}u + \lambda^2\,A(\alpha)^{-1}\mathrm{softmax}(\beta r^{(k)}),
\]
initialized with small random noise $r^{(0)}\sim\mathcal{N}(0,\sigma_{\mathrm{init}}^2 I)$ where $\sigma_{\mathrm{init}}=10^{-4}$.
We run up to 500 iterations or until $\|r^{(k+1)}-r^{(k)}\|_{\infty}<10^{-8}$.
The retrieval probabilities are then $\hat{r}^*=\mathrm{softmax}(\beta r^*)$.

\textbf{Alpha sweep for panel (d).}
To capture the smooth transitions across all $\lambda$ values uniformly, we augment the sweep from Section~\ref{app:exp:wta} with 100 uniformly spaced points on $[0,2]$.
Additionally, for each $\lambda>0$, we add a dense grid of 60 points near the theoretical $\alpha_{\mathrm{crit}}^{\lambda}$.
After removing duplicates (rounded to 4 decimals), this yields approximately 150--180 $\alpha$ values.

\textbf{Memory generation and trials.}
Each trial uses the same single-cluster generation protocol (Section~\ref{app:exp:wta}), with $N=50$, $d=10$, $\sigma=0.3$, and $\beta=5.0$.
We run 1000 independent trials per $\alpha$ (seeds $0,\ldots,999$), computing retrieval probabilities for all $\lambda$ values simultaneously by reusing $A(\alpha)^{-1}$ for efficiency.

\textbf{Reported quantity.}
For each $(\alpha,\lambda)$ pair, we report the mean peak retrieval probability
\[
\mathbb{E}\Bigl[\max_{i\in[N]} p^*_i(\alpha,\lambda)\Bigr]
\]
over the 1000 trials.
Panel (d) visualizes how the phase transition shifts leftward (to smaller $\alpha$) as $\lambda$ increases.

\subsection{Replication of Separability and Alignment}
\label{app:replicate_icl_geom}

In this section, we replicate the analytical framework introduced by \cite{yangCZI2025} to verify the geometric evolution of hidden states across our selected datasets (SST-2, AG-News, and TREC). Specifically, we track two metrics across the transformer layers: \emph{separability}, quantified by the predictive accuracy of a linear probe (logistic regression) trained on the hidden representations, and \emph{alignment}, measured via the cosine similarity between the hidden states and the target unembedding vectors. 

As shown in Fig.~\ref{fig:replication_geometry}, our results faithfully reproduce the empirical phenomenology of the original study. Across all three datasets, we observe a distinct two-phase progression: class separability emerges and peaks in the early-to-middle layers, whereas direct alignment with the unembedding matrix sharpens exclusively in the final layers. This timescale separation directly supports our architectural assumption that context processing (separability) settles well before query execution (alignment).

\begin{figure}[ht]
\centering
\begin{subfigure}{0.8\linewidth}
\centering
\includegraphics[width=\linewidth]{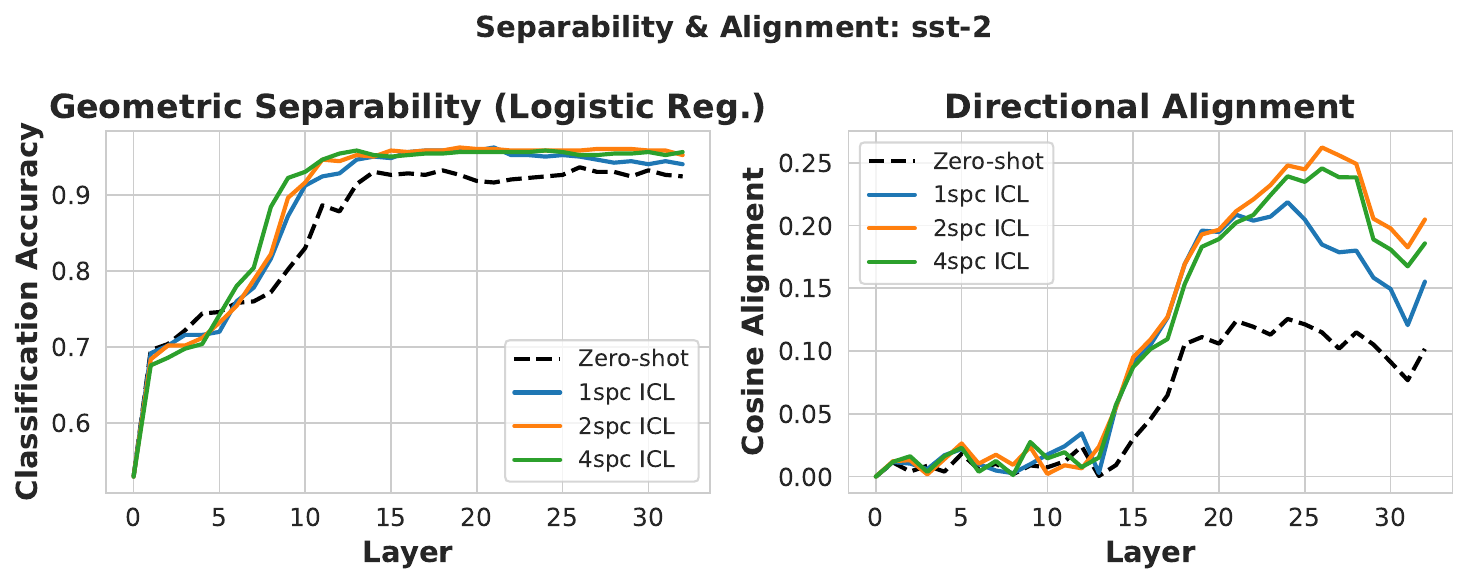}
\caption{SST-2}
\label{fig:rep_sst2}
\end{subfigure}

\vspace{0.4cm} 

\begin{subfigure}{0.8\linewidth}
\centering
\includegraphics[width=\linewidth]{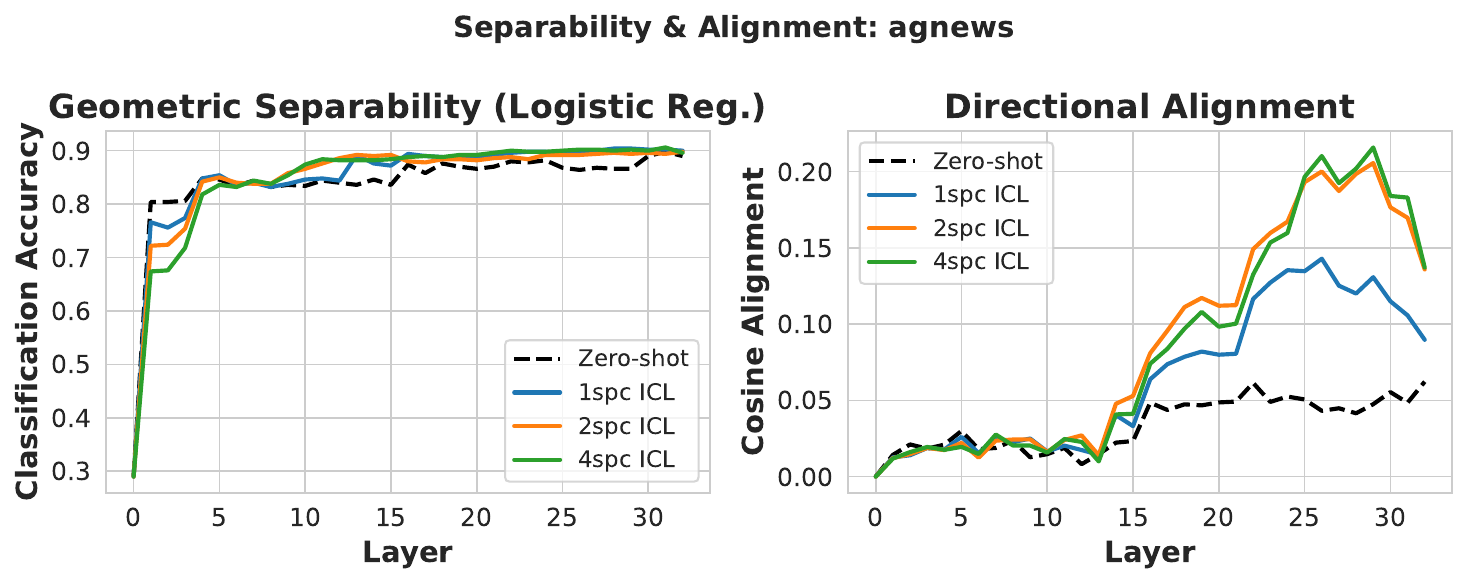}
\caption{AG-News}
\label{fig:rep_agnews}
\end{subfigure}

\vspace{0.4cm} 

\begin{subfigure}{0.8\linewidth}
\centering
\includegraphics[width=\linewidth]{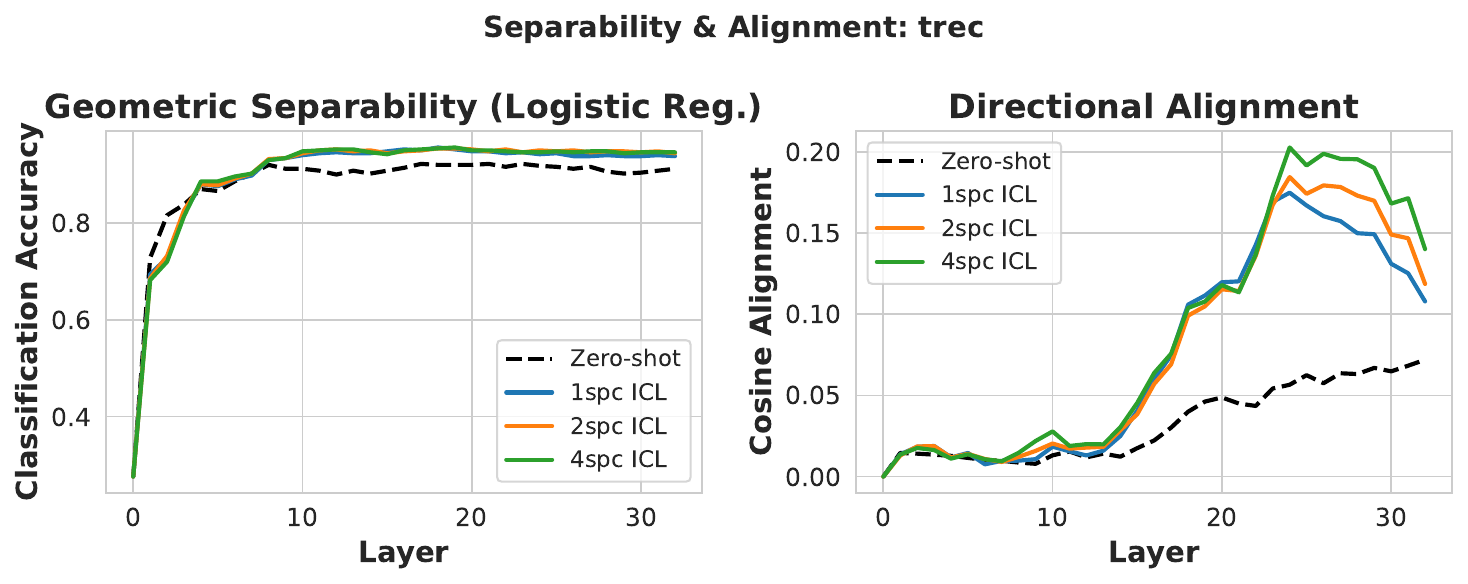}
\caption{TREC}
\label{fig:rep_trec}
\end{subfigure}
\caption{\textbf{Evolution of Separability and Alignment across layers.} Replication of the geometric progression showing early-layer linear separability (via logistic regression) followed by late-layer unembedding alignment across three classification datasets.}
\label{fig:replication_geometry}
\end{figure}

\subsection{Additional Datasets for ICL Classification: AG-News \& TREC}
\label{app:add_datasets}

Building upon the layer-sweep analysis presented in the main text for SST-2, we present the extended coupled retrieval experiments for the AG-News ($L=4$) and TREC ($L=6$) datasets. Fig.~\ref{fig:neff_layers_agnews} and Fig.~\ref{fig:neff_layers_trec} show the concentration phenomenon induced by the context, while Fig.~\ref{fig:retrieval_sweep_agnews} and Fig.~\ref{fig:retrieval_sweep_trec} display the retrieval accuracy across all permutations of context-extraction layers and zero-shot query layers under 1-shot and 4-shot settings. 

\begin{figure}[ht]
\centering
\includegraphics[width=0.8\linewidth]{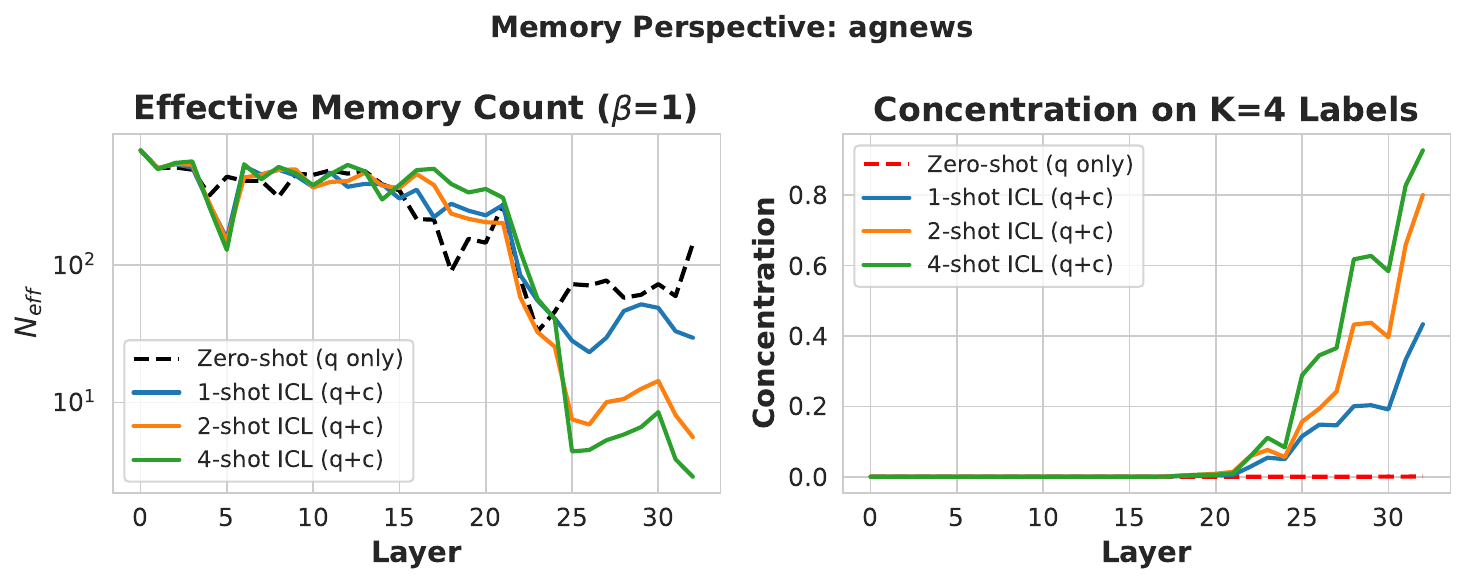}
\caption{\textbf{Native ICL processing collapses the memory space onto the
label set.} Effective number of active memories $N_\mathrm{eff}$ obtained
by decoding $h_{\mathrm{zero}}^{(\ell)}$ and $h_{\mathrm{ICL}}^{(\ell)}$
through the unembedding at each layer, shown across layers and shot
counts on AG-News.}
\label{fig:neff_layers_agnews}
\end{figure}

\begin{figure}[ht]
\centering
\includegraphics[width=0.8\linewidth]{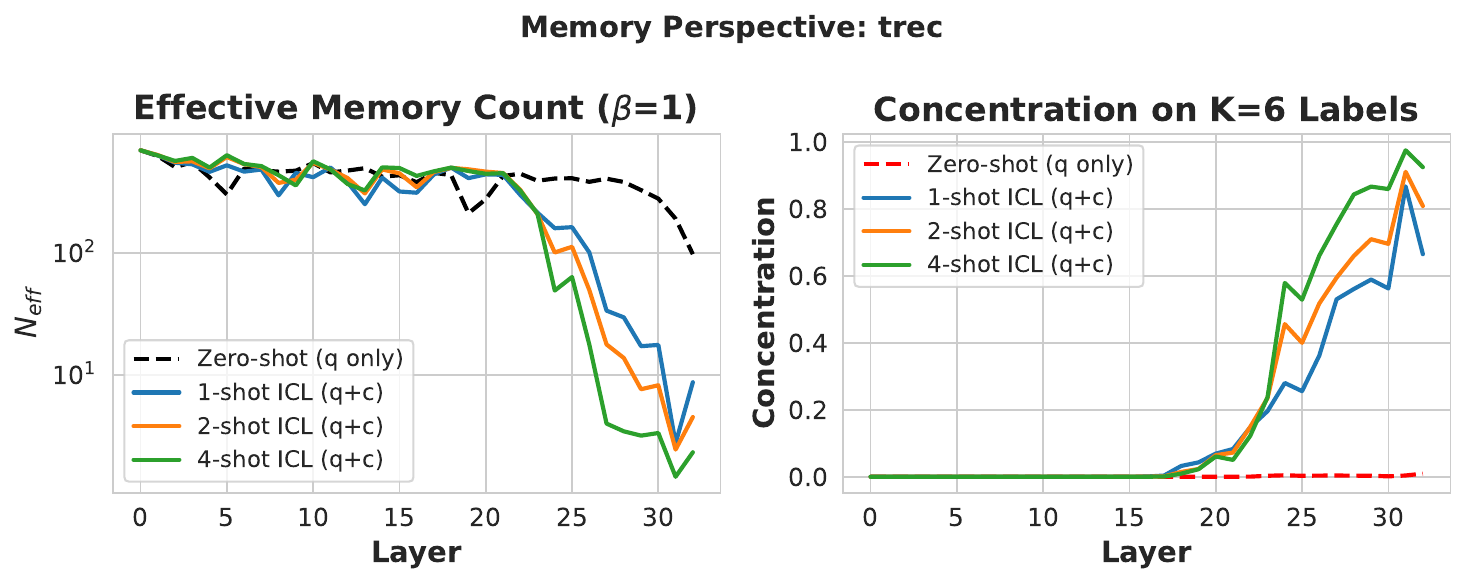}
\caption{\textbf{Native ICL processing collapses the memory space onto the
label set.} Effective number of active memories $N_\mathrm{eff}$ obtained
by decoding $h_{\mathrm{zero}}^{(\ell)}$ and $h_{\mathrm{ICL}}^{(\ell)}$
through the unembedding at each layer, shown across layers and shot
counts on TREC.}
\label{fig:neff_layers_trec}
\end{figure}

The trends across these multi-class datasets corroborate our primary binary classification findings. Despite dropping the second-order feedback loop to test a strict first-order approximation, the linear additive formulation recovers a decent chunk of full ICL performance. Note that the retrieval accuracies under our paradigm are much lower, both due to the fact the baseline ICL LLM accuracy is much lower, and due to the increased difficulty of the task (larger label set). Nevertheless, consistently across tasks, context vectors extracted from the later half of the network ($\ell \ge 20$) carry sufficient structural information to successfully reshape the energy landscape and drive accurate retrieval.

\begin{figure}[ht]
\centering
\begin{subfigure}{\linewidth}
\centering
\includegraphics[width=0.75\linewidth]{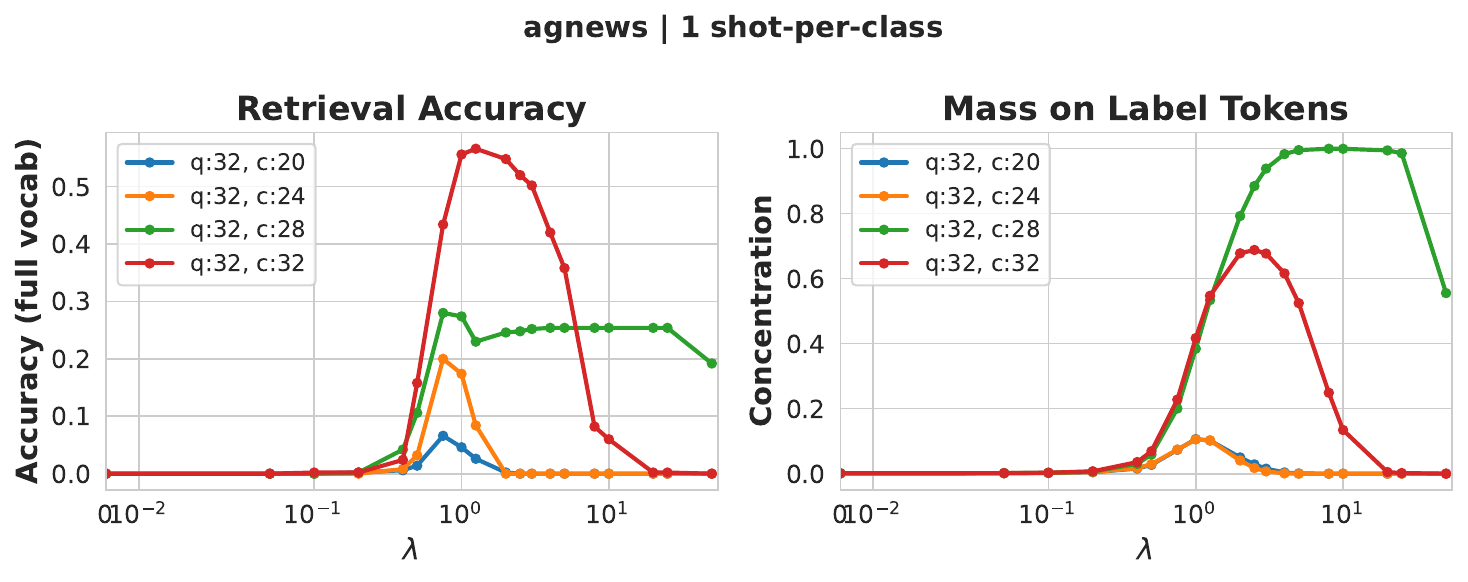}
\label{fig:agnews_retrieval_1shot}
\end{subfigure}
\vspace{0.2cm} 
\begin{subfigure}{\linewidth}
\centering
\includegraphics[width=0.75\linewidth]{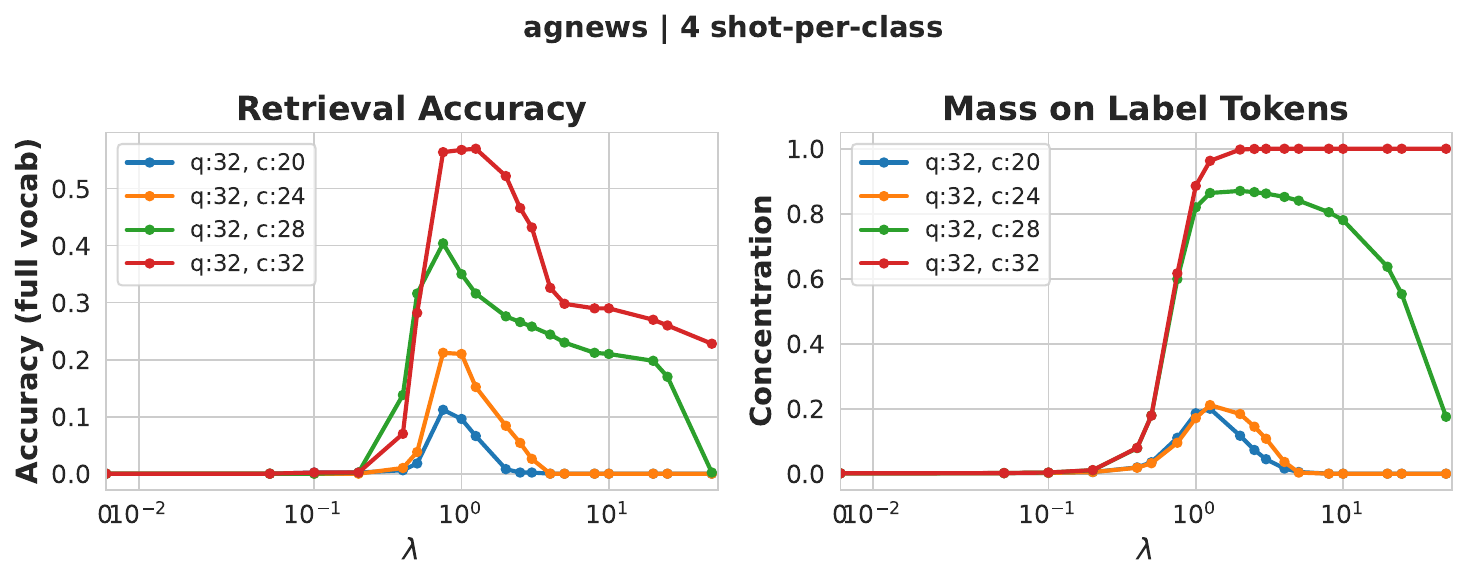}
\label{fig:agnews_retrieval_4shot}
\end{subfigure}
\caption{\textbf{Coupled retrieval across context-query layer combinations (AG-News).} Retrieval accuracy as a function of coupling strength $\lambda$, context extraction layer $\ell$, and number of shots (1 and 4) on the AG-News dataset.}
\label{fig:retrieval_sweep_agnews}
\end{figure}

\begin{figure}[ht]
\centering
\begin{subfigure}{\linewidth}
\centering
\includegraphics[width=0.75\linewidth]{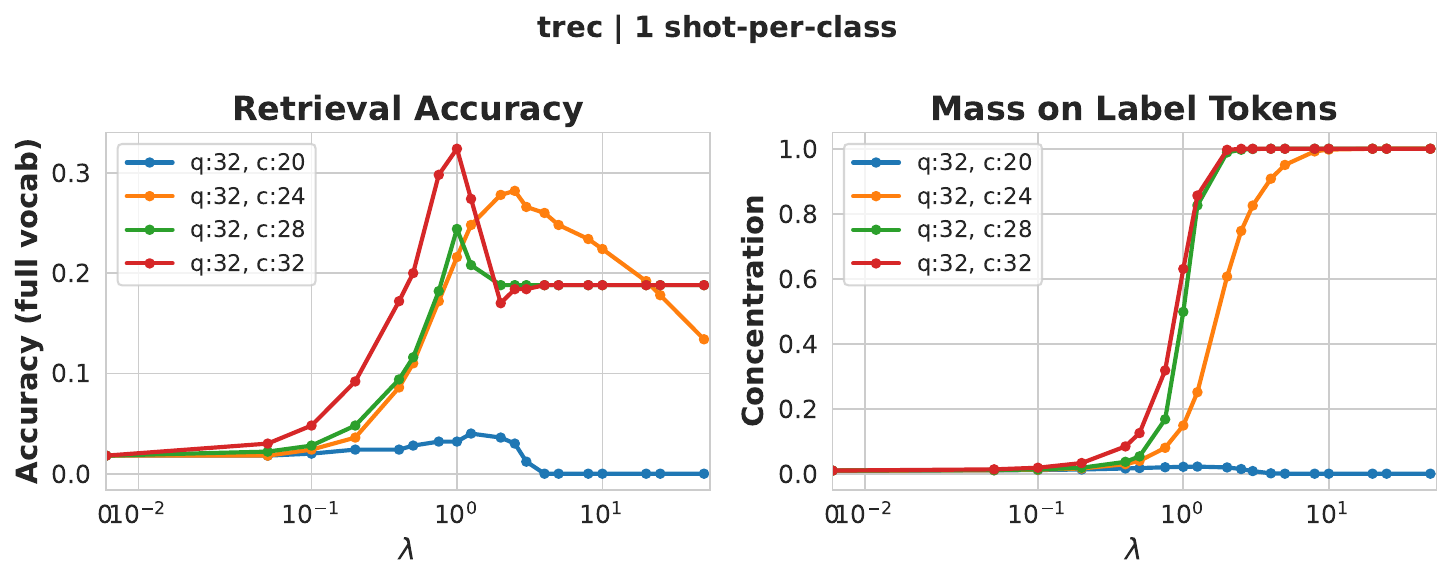}
\label{fig:trec_retrieval_1shot}
\end{subfigure}
\vspace{0.2cm} 
\begin{subfigure}{\linewidth}
\centering
\includegraphics[width=0.75\linewidth]{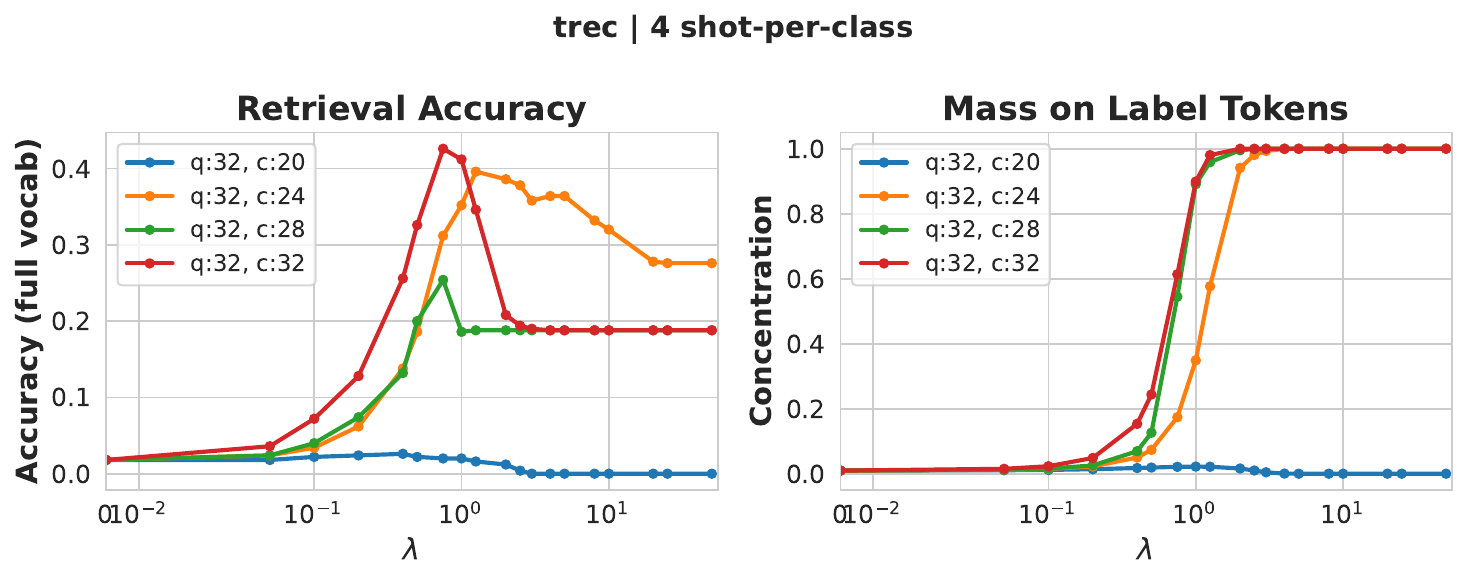}
\label{fig:trec_retrieval_4shot}
\end{subfigure}
\caption{\textbf{Coupled retrieval across context-query layer combinations (TREC).} Retrieval accuracy as a function of coupling strength $\lambda$, context extraction layer $\ell$, and number of shots (1 and 4) on the TREC dataset.}
\label{fig:retrieval_sweep_trec}
\end{figure}

\clearpage


\end{document}